\documentclass{article}
\usepackage[OT1]{fontenc}
\usepackage[utf8]{inputenc}
\usepackage{kpfonts}
\usepackage{graphicx}
\usepackage[dvipsnames]{xcolor}
\usepackage{amsthm}
\usepackage[margin=1.25in]{geometry}
\usepackage[sf,bf,small,raggedright,compact]{titlesec}
\usepackage{hyperref}
\definecolor{linkblue}{named}{MidnightBlue}
\hypersetup{colorlinks=true, linkcolor=linkblue,  anchorcolor=linkblue,
        citecolor=linkblue, filecolor=linkblue, menucolor=linkblue,
        urlcolor=linkblue}
\usepackage[capitalize]{cleveref}
\usepackage{paralist}
\usepackage[longnamesfirst,numbers,sort&compress]{natbib}
\usepackage[multiple]{footmisc}

\usepackage{amssymb}
\usepackage{stmaryrd}

\usepackage{listings,newtxtt}
\DeclareMathOperator{\interior}{int}

\newcommand{\R}{\mathbb{R}}

\newcommand{\xx}{\ensuremath{\protect{x}}}
\newcommand{\yy}{\ensuremath{\protect{y}}}

\newtheorem{thm}{Theorem}

\newtheorem{lem}{Lemma}
\crefname{lem}{Lemma}{Lemmata}
\newtheorem{cor}{Corollary}

\crefname{dm}{}{}
\creflabelformat{dm}{#2(DM#1)#3}
\crefname{dp}{}{}
\creflabelformat{dp}{#2(DP#1)#3}
\crefname{p}{}{}
\creflabelformat{p}{#2(#1)#3}
\crefname{pp}{}{}
\creflabelformat{pp}{#2(P#1)#3}
\crefname{bg}{}{}
\creflabelformat{bg}{#2[bg#1]#3}

\newcommand{\slcf}{straight-line crossing-free drawing}

\newcommand{\embedding}{crossing-free drawing}

\newcommand{\defin}[1]{\emph{\textcolor{Maroon}{#1}}}

\theoremstyle{definition}

\title{Free Sets in Planar Graphs: History and Applications\thanks{This research is partially funded by the grants from NSERC and University of Ottawa Research Chair.}}
\author{
  Vida Dujmović\thanks{Department of Computer Science and Electrical Engineering, University of Ottawa.}\and
  Pat Morin\thanks{School of Computer Science, Carleton University.}}
\date{}

\DeclareMathOperator{\fix}{fix}
\DeclareMathOperator{\move}{shift}
\newcommand{\N}{\mathbb{N}}

\begin{document}
\begin{titlepage}
\maketitle

\begin{abstract}
  A subset $S$ of vertices in a planar graph $G$ is a \defin{free set} if, for every set $P$ of $|S|$ points in the plane, there exists a \slcf\ of $G$ in which vertices of $S$ are mapped to distinct points in $P$.  In this survey, we review
  \begin{compactitem}
      \item several equivalent definitions of free sets, 
      \item and applications of free sets in graph drawing.  
  \end{compactitem}
  The survey concludes with a list of open problems in this still very active research area.
\end{abstract}
\end{titlepage}

\pagestyle{empty}
\tableofcontents
\newpage

\pagestyle{plain}
\section{Introduction}

In 2005, an online game called \defin{Planarity} became wildly popular\footnote{\url{https://en.wikipedia.org/wiki/Planarity}}.  In this single-player game, created by John Tantalo based on a concept by Mary Radcliffe, the player is presented with a straight-line drawing $\Gamma$ of a planar graph $G$.  Although the graph $G$ is planar, the drawing $\Gamma$ is not crossing-free; some pairs of edges cross each other, and an edge may even contain a vertex in its interior.  Since $\Gamma$ draws the edges of $G$ as straight-line segments, the drawing is completely determined by the vertex locations.  The job of the player is to move vertices in the drawing $\Gamma$ in order to obtain a crossing-free drawing of $G$. Since $G$ is planar, Fáry's Theorem guarantees that this is always possible.

The act of moving some vertices in $\Gamma$ to obtain a crossing-free drawing is referred to as \defin{untangling}. For a combinatorially inclined player, this naturally leads to the question:  What is the fewest number, $\move(\Gamma)$, of vertices that need to be moved in order to untangle $\Gamma$?\footnote{The Planarity game actually measures the amount of time the player takes to accomplish this task, not the number of vertices moved.} Equivalently, what is the maximum number, $\fix(\Gamma):=|V(G)|-\move(\Gamma)$ of vertices that can be kept fixed while untangling $G$?  For a planar graph $G$, let $\fix(G):=\min\{\fix(\Gamma):\text{$\Gamma$ is a straight-line drawing of $G$}\}$.
 For a family $\mathcal{F}$ of planar graphs, this defines a function $\fix_{\mathcal{F}}(n):=\min\{\fix(G):G\in\mathcal{F},\, |V(G)|=n\}$.\footnote{Technically $\fix_{\mathcal{F}}(n)$ is undefined if there are no $n$-vertex graphs in the family $\mathcal{F}$. We will ignore this, since all the graph families we consider are infinite and have at least one $n$-vertex member for each $n\in\N$.}

Determining the asymptotic growth of $\fix_{\mathcal{F}}(n)$ turns out to be a challenging problem, even for very simple classes $\mathcal{F}$.  In 1998,  Mamoru~Watanabe had already asked this problem for the class $\mathcal{C}$ of cycles and we now know that $\fix_{\mathcal{C}}(n)\in \Omega(n^{2/3})\cap O((n\log n)^{2/3})$ but neither the lower bound, due to \citet{pach.tardos:untangling}, nor the upper bound, due to \citet{cibulka:untangling}, is easy to prove.  For the class $\mathcal{G}$ of planar graphs, the problem has been studied since \citet{pach.tardos:untangling} asked for a polynomial lower bound over $20$ years ago. The best known bounds to date are $\fix_{\mathcal{G}}(n)\in \Omega(n^{1/4})\cap O(n^{0.4948})$ \cite{bose.dujmovic.ea:untangling,kang.pikhurko.ea:untangling,goaoc.kratochvil.ea:untangling,cano.toth.ea:upper}.

In this survey, we study various kinds of vertex subsets of a planar graph $G$: proper-good sets, collinear sets, free-collinear sets, and free sets.  Each of these definitions is, at first glance, more stringent than the one that precedes it. Indeed, it follows immediately from definitions that every free set is a free-collinear set, every free-collinear set is a collinear set, and every collinear set is a proper-good set. The definition of a proper-good set is the most relaxed, which makes it easy to find large proper-good sets, even by hand using a pencil and a paper that contains a \embedding\ of $G$.  At the other extreme, free sets satisfy a very strong property that says we can place the vertices of a free set on any given pointset and then find locations for the other vertices of $G$ that result in a non-crossing straight-line drawing of $G$. This makes free sets incredibly useful for many applications, including untangling.

A key result in this area, discovered over roughly 15 years, is that all of these definitions are equivalent: A vertex subset in a planar graph is proper-good if and only if it is collinear, if and only if it is free-collinear, if and only if it is free.  Meanwhile, proper-good sets, collinear sets, free-collinear sets, and free sets were being used to resolve problems in graph drawing and related areas. With the benefit of hindsight, we attempt to organize and present two decades of results on free sets and their applications in a way that makes them as easy to understand and as useful as possible.

The rest of this survey is structured as follows: In \cref{four_definitions} we present four definitions of free sets and explain why these four definitions are equivalent. In \cref{large_free_sets} we present upper and lower bounds on the size of free sets in planar graphs and various subclasses of planar graphs.  In \cref{applications} we describe applications of free sets to a variety of graph drawing problems. In \cref{one_bend_section}, we introduce a one-bend variant of free sets. \Cref{conclusion} concludes with a list of open problems and directions for further research.

\section{Four Definitions of Free Sets}
\label{four_definitions}

For definitions of standard graph theoretic terms and notations used in this survey (such a treewidth, $k$-trees, independent set, etc.) the reader is referred to the textbook by \citet{Diestel5}.  Throughout this survey, a \defin{graph} is undirected, has no self-loops, and has no parallel edges, unless specified otherwise.

A \defin{drawing} of a graph is a representation of the graph in which each vertex $v$ is represented by a distinct point in the plane and each edge $vw$ is represented by a simple open curve in the plane, the closure of which has $v$ and $w$ as endpoints.  When discussing a particular drawing, we do not distinguish between a vertex and the point that represents it, or an edge and the curve that represents it.  A drawing is a \defin{straight-line} drawing if each edge is an open line segment.   In a drawing, two edges may have a non-empty intersection, and an edge may contain a vertex that is not one of its endpoints.  A drawing is a \defin{\embedding} if this does not occur: any two distinct edges are disjoint and no vertex is contained in any edge.\footnote{What we call a \embedding\ is sometimes called a \defin{plane} drawing or an embedding (in the plane). What we call a \slcf\ is sometimes called a \defin{geometric embedding} or \defin{Fáry embedding}.} 

Recall that for any $n$-vertex planar graph $G_0$, on at least $3$ vertices, it is possible to add edges to $G_0$ to obtain a planar graph $G$ with $3(n-2)$ edges.  The resulting graph $G$ is a \defin{triangulation}: in any \embedding\ of this graph, each of the faces has exactly three edges on its boundary. A \defin{near-triangulation} is a biconnected embedded graph whose inner faces are all \defin{triangles}, i.e., bounded by three edges.  The \defin{dual} $G^*$ of an embedded graph $G$ is a graph whose vertices are the faces of $G$ that contains an edge between two vertices $f$ and $g$ if these two faces of $G$ have a common edge on their boundaries.  For a planar graph $G$ with outer face $f_0$, the \defin{weak dual} $G^+$ of $G$ is defined as $G^+:=G^*-f_0$, the graph obtained from the dual $G^*$ by removing the vertex corresponding to the outer face $f_0$ of $G$. A graph is \defin{outerplanar} if it has a \embedding\ with all its vertices on the outer face. A graph equipped with such a \embedding\ is an \defin{outerplane} graph. A \defin{chord} of a cycle $C$ is an edge between two vertices of $C$ that are not adjacent in $C$.

Under our definitions, a \defin{planar} graph is a graph that has a \embedding. \defin{Fáry's Theorem} \cite{wagner:bemerkungen,fary:on,stein:convex} asserts that a graph is planar if and only if it has a \slcf.  Today, we use Fáry's Theorem without thinking about it, but it is worthwhile spending a moment reflecting on what would happen if Fáry's Theorem were not true.  There would be two kinds of planar graphs, topological and straight-line, with the latter being a special case of the former.  Any result proven for straight-line planar graphs would need a separate proof for topological planar graphs (assuming it was also true for topological planar graphs).  Edge contractions, which obviously preserve topological planarity and are commonly used in inductive proofs for planar graphs, would be off-limits for proofs about straight-line planar graphs.  

Indeed, there are a number of questions in which the difference between (not necessarily crossing-free) drawings and straight-line drawings of graphs is profound.  A famous example of this is Conway's \defin{Thrackle Conjecture} \cite{woodall:thrackles}, which asserts that if $G$ has a drawing in which each pair of edges share exactly one point (possibly a common endpoint) then the number of edges of $G$ is no more than the number of vertices of $G$.  This conjecture has been open since the 1960s but the straight-line version of the question (the \defin{Linear Thrackle Problem}) has an easy solution, first observed by Erdős \cite{erdos:on}.

Let $\Gamma$ be a \embedding\ of a planar graph $G$.  A simple closed curve\footnote{A \defin{simple closed curve} $C$ is a continuous function $C:[0,1]\to\mathbb{R}^2$ such that $C(0)=C(1)$ and $C(a)\neq C(b)$ for all $0\le a< b<1$.} $C:[0,1]\to\R^2$ is \defin{good} with respect to $\Gamma$ if $C(0)$ is not contained in any edge or vertex of $\Gamma$. The curve $C$ is \defin{proper} with respect to $\Gamma$ if, for each edge $e$ of $\Gamma$, the intersection of $C$ and the closure of $e$ is either empty, consists of a single point (possibly an endpoint of $e$), or consists of the entire edge $e$.  In essence, a curve that is proper-good with respect to $\Gamma$ behaves the way a line would behave if $\Gamma$ were a straight-line drawing.

Let $S:=(v_1,\ldots,v_s)$ be an ordered subset (that is, a sequence of distinct elements) of vertices in a planar graph $G$.
\begin{enumerate}
    \item $S$ is a \defin{proper-good set} \cite{dalozzo.dujmovic.ea:drawing} if there exists a \embedding\ $\Gamma$ of $G$, a curve $C:[0,1]\to\R^2$ that is proper and good with respect to $\Gamma$, and $0 < x_1<\cdots< x_s<1$ such that $C(x_i)$ is the location of $v_i$ in $\Gamma$, for each $i\in\{1,\ldots,s\}$.  In other words, we encounter the vertices of $S$, in order, while traversing $C$.

    \item $S$ is a \defin{collinear set} \cite{DBLP:conf/wg/RavskyV11} if there exists a \slcf\ $\Gamma$ of $G$ and $x_1<\cdots<x_s$ such that $(x_i,0)$ is the location of $v_i$ in $\Gamma$, for each $i\in\{1,\ldots,s\}$.

    \item $S$ is a \defin{free-collinear set} \cite{DBLP:conf/wg/RavskyV11} if, for any $x_1<\cdots<x_s$, there exists a \slcf\ $\Gamma$ of $G$ such that $(x_i,0)$ is the location of $v_i$ in $\Gamma$, for each $i\in\{1,\ldots,s\}$.

    \item $S$ is a \defin{free set} if, for any $x_1<\cdots<x_s$ and \emph{any} $y_1,\ldots,y_s$, there exists a \slcf\ $\Gamma$ of $G$ such that $(x_i,y_i)$ is the location of $v_i$ in $\Gamma$, for each $i\in\{1,\ldots,s\}$.
\end{enumerate}

We note that our definitions deviate from much of the literature by treating $S$ as an ordered set. Later, we may say that some (unordered) vertex subset $S$ of $G$ is a proper-good, collinear, free-collinear, or free set.  In these cases, we mean that there is some permutation of $S$ that defines an ordered set that satisfies the relevant definition. When we want to emphasize the difference between these two, we will say that (an ordered set) $S$ is an \defin{ordered free set} or that (a set) $S$ is an \defin{unordered free set}.  This distinction becomes important when discussing two or more graphs with the same vertex set.  A set $S$ of vertices may be an unordered free set in two graphs $G_1$ and $G_2$, but there may be no permutation of $S$ that is an ordered free set in $G_1$ and in $G_2$.

A few properties of free sets are immediate from the definition, and we will use them throughout.  If $S$ is an ordered free set in a graph $G$, then:
\begin{compactitem}
  \item the reversal of $S$ is an ordered free set in $G$;
  \item $S$ is an ordered free set in any subgraph of $G$ that spans $S$; and
  \item any subsequence of $S$ is an ordered free set in $G$.
\end{compactitem}

Unordered collinear sets and unordered free-collinear sets were defined first by \citet{DBLP:conf/wg/RavskyV11} and also appear implicitly in the work of \citet{bose.dujmovic.ea:untangling} on untangling.\footnote{Although they don't use this terminology, \citet{bose.dujmovic.ea:untangling} establish that $\fix_{\mathcal{G}}(n)\in \Omega(n^{1/4})$ by showing that every planar graph has a free collinear set of size $\Omega(n^{1/2})$.} In fact, \citet{DBLP:conf/wg/RavskyV11} posed the equivalence between (unordered) collinear sets and (unordered) free-collinear sets as an open question, and conjectured a negative answer. \Cref{equivalence} below disproves that conjecture. Proper-good sets (as unordered sets) were introduced first by \citet{dalozzo.dujmovic.ea:drawing}.

Before continuing, we make a remark about pointsets whose \xx-coordinates are not all distinct, since the definition of (ordered) free set appears to disallow these. However, we claim that for any unordered free set $S$ in a planar graph $G$, and \emph{any} set $P$ of $|S|$ points in the plane, there exists a \slcf\ of $G$ in which each vertex in $S$ is mapped to some point in $P$. If all the points in $P$ have distinct \xx-coordinates then the claim follows immediately from the definition of ordered free set.  If this is not the case, then a slight rotation of $P$ gives a pointset $P'$ in which each point has a distinct \xx-coordinate.  Now $G$ has a \slcf\ $\Gamma'$ in which each vertex in $S$ maps to a point in $P'$.  Applying the inverse rotation to $\Gamma'$ gives a drawing $\Gamma$ in which each vertex in $S$ maps to a point in $P$.

Next we turn to the relationship between free sets and the untangling problem in the introduction.  In the untangling problem, we are given a straight line drawing $\Gamma$ of $G$ (with crossings) and the goal is to find a large set $F\subseteq V(G)$ and a \slcf\ $\Gamma'$ of $G$ such that $\Gamma'(v)=\Gamma(v)$ for each $v\in F$.  Note that it is not sufficient to choose $F$ to be a free set in $G$ for the following reason: The definition of a free set would only guarantee that we can find a \slcf\ $\Gamma'$ such that $\{\Gamma'(v):v\in F\}=\{\Gamma(v):v\in F\}$.  That is, $\Gamma'$ draws the vertices of $F$ on the same point set $P:=\{\Gamma(v):v\in F\}$ in both drawings, but may change the permutation that maps the vertices in $F$ to points in $P$.  The following lemma shows that free sets in $G$ are nevertheless useful in obtaining lower bounds on $\fix(G)$. The proof of this lemma also illustrates the usefulness of definitions based on ordered sets. Another example, discussed in \cref{applications}, is the problem of simultaneous embedding with mapping.

\begin{lem}\label{free_to_fix}
    Let $G$ be a planar graph that has a free set of size at least $k$.  Then $\fix(G) \ge \sqrt{k}$.
\end{lem}

\begin{proof}
    Let $S$ be an ordered free set in $G$ of size $k$. Fix any straight-line drawing $\Gamma$ of $G$ and, without loss of generality (by a slight rotation, if necessary) assume that no two vertices of $\Gamma$ have the same \xx-coordinate.  Let $(S,\prec)$ be the partial order in which $v\prec w$ if and only if $v$ appears before $w$ in $S$ and the \xx-coordinate of $v$ in $\Gamma$ is less than the \xx-coordinate of $w$ in $\Gamma$.

    By Dilworth's Theorem, $(S,\prec)$ contains a chain of length at least $\sqrt{k}$ or $(S,\prec)$ contains an antichain of size at least $\sqrt{k}$.  First, suppose $(S,\prec)$ contains a chain $v_1\prec\cdots\prec v_\ell$ of length $\ell\ge \sqrt{k}$.  Then $(v_1,\ldots,v_\ell)$ is a free set since it is a subsequence of $S$.  For each $i\in\{1,\ldots,\ell\}$, let $(x_i,y_i)$ be the coordinates given to $v_i$ in $\Gamma$.  By the definition of free set, $G$ has a \slcf\ in which $v_i$ is mapped to $(x_i,y_i)$, for each $i\in\{1,\ldots,\ell\}$.  Therefore, $\fix(\Gamma)\ge \ell\ge \sqrt{k}$.

    To deal with the case where $(S,\prec)$ contains an antichain of size at least $\sqrt{k}$, we use the fact that the ordered set obtained by reversing the order of $S$ is also a free set. Let $(S,\prec')$ be the partial order on the elements of $S$ in which $v\prec' w$ if and only if $v$ appears \emph{after} $w$ in $S$ and the \xx-coordinate of $v$ in $\Gamma$ is less than the \xx-coordinate of $w$ in $\Gamma$.  The elements of any antichain in $(S,\prec)$ form a chain in $(S,\prec')$. Now we proceed as in the previous paragraph to show that $\fix(\Gamma)\ge\sqrt{k}$.
    Therefore $\fix(\Gamma)\ge\sqrt{k}$ for any straight-line drawing $\Gamma$ of $G$.  Therefore $\fix(G)\ge \sqrt{k}$.
\end{proof}

For most classes $\mathcal{X}$ of planar graphs, including trees, outerplanar graphs, and all planar graphs, the best known lower bounds (and for some classes asymptotically optimal bounds) for $\fix_\mathcal{X}(n)$ can be obtained by an application of \cref{free_to_fix}, along with a result on free sets in graphs in the class $\mathcal{X}$ \cite{goaoc.kratochvil.ea:untangling, dalozzo.dujmovic.ea:drawing, bose.dujmovic.ea:untangling, DBLP:conf/wg/RavskyV11}.  The sole exception is the class $\mathcal{C}$ of cycles, for which \citet{cibulka:untangling} uses repeated applications of the Erd\H{o}s-Szekeres Theorem to obtain a lower bound of $\Omega(n^{2/3})$.

As discussed in the introduction, the following theorem provides the main motivation for this survey.

\begin{thm}\label{equivalence}
    Let $G$ be a planar graph and let $S$ be an ordered subset of $V(G)$. Then the following statements are equivalent:
    \begin{compactenum}
        \item $S$ is an (ordered) proper-good set.
        \item $S$ is an (ordered) collinear set.
        \item $S$ is an (ordered) free-collinear set.
        \item $S$ is an (ordered) free set.
    \end{compactenum}
\end{thm}

Note that, unlike the other three definitions, the definition of proper-good set is purely combinatorial/topological; it has no mention of straight-lines or vertex coordinates.   In this respect, \cref{equivalence} plays a role similar to Fáry's Theorem.  Using the proper-good set definition allows purely combinatorial methods to be used in the search for free sets.  For any proper-good set, there is a short easily-verifiable certificate (a \embedding\ of $G$ and a description of the curve $C$) that $S$ is indeed a proper-good set.  In contrast, in order to show directly that $S$ is a free-collinear set (or a free set) involves proving a fact that must hold for all choices of $x_1,\ldots,x_k$ (or $(x_1,y_1),\ldots,(x_k,y_k)$, respectively).  Working with this definition directly would involve (at the very least) a significant amount of linear algebra.

The proof of \cref{equivalence} is spread across a few different works \cite{dalozzo.dujmovic.ea:drawing,dujmovic.frati.ea:every,bose.dujmovic.ea:untangling}, which we describe below.  First, we remark that it is immediate from the definitions that every free set is a free-collinear set and every free-collinear set is a collinear set.  It is also easy to see that every collinear set is a proper-good set by taking a \slcf\ $\Gamma$ of $G$ with the vertices of the (ordered) collinear set $S$ on the \xx-axis in increasing order of \xx-coordinate.  Let $s$ be a line segment that is contained in the \xx-axis and that contains every intersection between $\Gamma$ and the \xx-axis in its interior.  Then $s$ is proper with respect to $\Gamma$. Finally, make $s$ into a closed curve by joining the endpoints of $s$ with a curve $C$ that completely avoids $\Gamma$.

We now address the other directions.

\paragraph{Every proper-good set is a collinear set:}
The proof that every proper-good set is a collinear set is due to \citet{dalozzo.dujmovic.ea:drawing} and is illustrated in \cref{one_to_two}. For each edge of $G$ that crosses $C$, the proof introduces a dummy vertex (white vertices in \cref{one_to_two}(a)) that splits the edge into two edges with a common vertex on $C$.  The resulting graph $G'$ then has three types of vertices: a set $X$ of inner vertices (contained in the interior of $C$), a set $Y$ of boundary vertices (contained in $C$) and a set $Z$ of outer vertices (the vertices in the exterior of $C$). These steps are depicted in \cref{one_to_two}(a).

\begin{figure}
    \centering
    \begin{tabular}{c@{}c@{}c}
    \includegraphics[page=1]{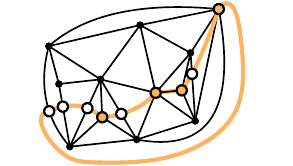} &
    \includegraphics[page=2]{figs/one_to_two} &
    \includegraphics[page=3]{figs/one_to_two} \\
    (a) & (b) & (c)
    \end{tabular}
    \caption{Showing that every proper-good set is a collinear set. Dummy vertices are hollow circles, and elements of $S$ are filled circles.}
    \label{one_to_two}
\end{figure}

The proof then uses Tutte's Convex Embedding Theorem \cite{tutte:how} twice:  Once to find a \slcf\ of $G'[X\cup Y]$ with all vertices of $X$ below the \xx-axis and all vertices of $Y$ on the \xx-axis, and a second time to find a \slcf\ of $G'[Y\cup Z]$ with the vertices of $Z$ above the \xx-axis.  This gives a \slcf\ $\Gamma'$ of $G'$, as depicted in  \cref{one_to_two}(b).  More precisely, the proof applies a variant of Tutte's Convex Embedding Theorem to a supergraph of $G''$ of $G'[X\cup Y]$  (respectively, $G'[Y\cup Z]$) whose outer face consists of a cycle that contains all the vertices of $Y$ plus a newly-introduced vertex $v$.\footnote{Tutte's Convex Embedding Theorem, as it appears in \cite{tutte:how} requires that the vertices of the outer face be mapped to a \emph{strictly} convex polygon.  However, the proof holds in the current setting because (i)~the vertices of the outer face of $G''$ induce a cycle and (ii) for each vertex $v$ of $G'[X\cup Y]$ not on the outer face, $N_G(v)$ does not consist entirely of collinear vertices on the outer face.} (The extra vertices and edges of these supergraphs are grey in \cref{one_to_two}(b).)

The \slcf\ $\Gamma'$ is almost a \slcf\ of $G$ except that the edges that properly intersect $C$ are represented by two line segments, one below and one above the \xx-axis. The proof then applies a result of \citet{pach.toth:monotone} (Theorem 1.2 in \cite{pach.toth:monotone}) that shows the edges in this \embedding\ can be straightened without changing the \yy-coordinate of any vertex, as depicted in  \cref{one_to_two}(c).  In particular, in the resulting \slcf\ of $G$, the vertices of $G$ intersected by $C$ remain on the \xx-axis, as required.

A few notes on the proof are in order.  First, we recall that the proof of Tutte's Convex Embedding Theorem is mostly algebraic: it involves setting up a system of linear equations whose solution determines the coordinates of the vertices in the \slcf.  The proof that every collinear set is a free-collinear set, sketched below, has a similar algebraic part.  Second, we note that the preceding proof \emph{almost} proves that every proper-good set is a \emph{free}-collinear set.  Tutte's theorem allows us to choose the \xx-coordinates of the vertices in $Y$, including the vertices in $S$.  However, the straightening step, which converts $\Gamma'$ into $\Gamma$ by straightening the edges that cross the \xx-axis, changes the \xx-coordinates of the vertices.

\paragraph{Every collinear set is a free-collinear set:}  The proof that every collinear set is a free-collinear set is due to \citet{dujmovic.frati.ea:every}.  Prior to this, \citet{dalozzo.dujmovic.ea:drawing} had established the equivalence of collinear and free-collinear sets for the class of planar $3$-trees and asked again if this equivalence was true in general.  In subsequent years, this repeatedly led to the question (originally posed by \citet{DBLP:conf/wg/RavskyV11})  ``is every collinear set free?''.
The proof of \citet{dujmovic.frati.ea:every} has two main parts, which we outline below.

If $S:=(v_1,\ldots,v_s)$ is a collinear set in $G$ then it is also a proper-good set in $G$. By adding one vertex in the outer face of $G$ and adding (not-necessarily straight-line) edges, we may assume that $G$ is a triangulation (and $S$ is still a proper-good set). As discussed in the previous section, the fact that $S$ is a proper-good set implies that $G$ is a collinear set, so $G$ has a \slcf\ $\Gamma$ in which the vertices of $S$ appear, in order, on the \xx-axis.  The goal is to show that, for any $x_1<\cdots<x_{s}$, the vertex $v_i$ can be moved to $(x_i,0)$, for each $i\in\{1,\ldots,s\}$, such that the resulting drawing is a \slcf.  In the following, when we say that an edge or a face \defin{properly intersects} the \xx-axis, we mean that the edge or the face contains points that are strictly above and points that are strictly below the \xx-axis.

The first part of the proof is combinatorial.  In this part, a sequence of modifications is done to $G$ with the end goal of obtaining a triangulation $G'$ that has a (no longer straight-line) \embedding\ $\Gamma'$ with the following properties (see \cref{a_graph}(c)):

\begin{compactenum}[(P1)]
  \item Every vertex in $S$ is a vertex of $G'$ and is mapped to the same location (on the \xx-axis) in both $\Gamma$ and $\Gamma'$.\label[pp]{s_preserving}

  \item The closure of each edge of $\Gamma'$ intersects the \xx-axis in at most one point. \label[pp]{x_proper}

  \item For each vertex $v\in S$ there are exactly two faces of $G'$ incident to $v$ that properly intersect the \xx-axis.

  \item If an edge $vw$ of $G'$ has no endpoint in $S$ then $vw$ properly crosses the \xx-axis or $vw$ is on the boundary of two faces $vwa$ and $vwb$ and the closure of each of these faces intersects the \xx-axis.
\end{compactenum}

Note that \cref{s_preserving} and \cref{x_proper} imply that no edge of $G'$ has both endpoints in $S$.  For reasons discussed below, we may assume that $G$ has no separating triangles and that no separating triangles are created during these operations.

To obtain $G'$, we apply three types of operations, two of which are illustrated in \cref{a_graph}.
\begin{inparaenum}[(1)]
    \item If some edge with no endpoint in $S$ is incident on two faces, neither of which properly intersect the \xx-axis, then this edge is contracted (the orange edge in \cref{a_graph}(a) meets those conditions and its contraction results in an embedded graph depicted in \cref{a_graph}(b)). This operation eliminates these two faces.
    \item If some edge with no endpoints in $S$ is incident on one face whose closure is disjoint from the \xx-axis, whose other incident face properly intersects the \xx-axis, and meets some other technical conditions,\footnote{These technical conditions ensure that, after drawing $G'$ the union of the two faces incident on the flipped edge is a convex quadrilateral.} then this edge is flipped (the lilac edge in \cref{a_graph}(b) meets the necessary conditions and its flipping results in an embedded graph depicted in \cref{a_graph}(c)). This edge flip replaces the two faces incident to the original edge with two faces that each properly intersect the \xx-axis.\footnote{The assumption that $G$ has no separating triangles ensures that this flipping operation does not introduce parallel edges, i.e., each edge of $G$ is flippable.}
    \item Finally, any edge with both endpoints in $S$ is flipped.  After this flip, the flipped edge and the two faces incident to it properly intersect the \xx-axis.
    \end{inparaenum}
 These operations (contractions and flips) are done exhaustively, one by one, in any order. The operations maintain that $S$ is a proper-good set in $G'$. As discussed above, either of these operations may introduce a separating triangle in $G'$. How this is handled is explained below.
Once these operations are done exhaustively, removing the edges of $G'$ both of whose endpoints are either strictly above or strictly below  \xx-axis (see the green edges in \cref{a_graph}(c)) yields a structure, called an \defin{A-graph} (see \cref{a_graph}(d)), in which every face is either a quadrilateral whose four edges each intersect the \xx-axis or have an endpoint on the \xx-axis or a triangle with vertices above, on, and below the \xx-axis.  The main task now is to find a \slcf\ of this $A$-graph so that the vertices of $S$ are moved to the specified locations on the \xx-axis. Once we achieve that, the edge-removal, flipping, and contracting operations are easily undone without changing the locations of vertices in $S$, resulting in \slcf\ of $G$ with the vertices of $S$ at their specified locations on the \xx-axis.

\begin{figure}
    \centering
    \begin{tabular}{cc}
        \includegraphics[page=2]{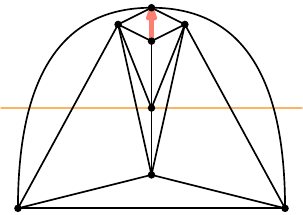} &
        \includegraphics[page=4]{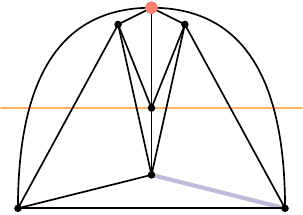} \\
        (a) & (b) \\
        \includegraphics[page=6]{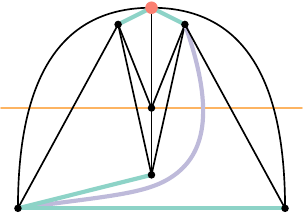} &
        \includegraphics[page=7]{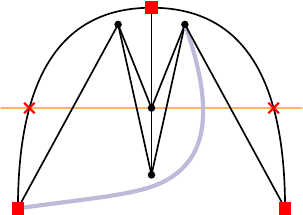} \\
        (c) & (d)
    \end{tabular}
    \caption{Modifying a triangulation to obtain an $A$-graph.}
    \label{a_graph}
\end{figure}

Finding a \slcf\ of this $A$-graph so that the vertices of $S$ are placed at the specified locations on the \xx-axis is the goal of the second part of the proof, which involves a combination of linear algebra and graph theory. The proof establishes a very strong statement about the A-graph $G'$.  The proof shows that it is possible to obtain a \slcf\ of $G'$ in which the intersection point of the closure of each edge with the \xx-axis is specified, subject to the ordering constraints imposed by the \embedding\ of $G'$. Note that this implies that $S$ is a free-collinear set in $G'$ since each vertex in $S$ is incident to an edge of $G'$ (recall that no edge incident to a vertex of $S$ is ever contracted and no edge incident to a vertex of $S$ is flipped unless both of its endpoints are on the \xx-axis).  Each vertex $v_i \in S$ has at least three incident edges in $G'$.  By specifying that the closure of each of these edges must intersect the \xx-axis at $(x_i,0)$, we ensure that $v_i$ is mapped to $(x_i,0)$.

In addition to specifying the location where each edge crosses the \xx-axis, the locations of the three vertices $x,y,z$ on the convex hull of this \embedding\ 
 (the vertices marked with \textcolor{red}{\rule{1.5ex}{1.5ex}} in \cref{a_graph}(d)) can be chosen, provided that the choice agrees with the choice of intersection points chosen for the edges induced by $x,y,z$ (the intersections marked with \textcolor{red}{$\times$} in \cref{a_graph}(d)). This latter requirement justifies the assumption that the graph does not contain separating triangles: The subgraph in the closure of the exterior of a separating triangle $xyz$ can be drawn inductively. This determines the locations of $x$, $y$, and $z$, which are then specified when inductively drawing the graph in the closure of the interior of the separating triangle.

The proof works by assigning slopes to the edges in $G'$.  To avoid vertical edges, it is helpful to rotate the coordinate system by $90$ degrees so that the roles of the \xx- and \yy-axes swap.  To prove this stronger claim about A-graphs, the proof sets up a system $M$ of linear equations whose variables are the slopes of the edges of $G'$ and whose coefficients are determined by the (given) intersection location of each edge with the \yy-axis.  Most of the equations in $M$ correspond to the fact that three or more edges incident to a common vertex $v$ must intersect in a single point (the location of $v$).\footnote{Some additional constraints, called proportionality constraints, are placed on the slopes of edges incident to vertices on the \yy-axis.}

The hardest step in this proof is to show that the linear system $M$ has a (unique) solution.  Here, the proof leverages the fact that $S$ is a proper-good set in $G'$ and, therefore, a collinear set in $G'$.  Therefore, $G'$ has a \slcf\ $\Gamma_0$ in which the vertices of $S$ appear, in order, on the \yy-axis.  From this \slcf\ $\Gamma_0$, one can read off a system $M_0$ of linear equations that has the same structure as $M$ but with different coefficients (that are determined by the \yy-coordinates of vertices in $S$ in $\Gamma_0$).  The system $M_0$ has a solution, namely the solution given by the slopes of the edges in $\Gamma_0$.  The proof first shows that this solution to $M_0$ is unique. The final step involves continuously modifying the coefficients in $M_0$ to obtain a continuum of linear systems $M_t$, $0\le t\le 1$, where $M_1=M$.  Using the fact that $M_0$ has a unique solution as a starting point, it is possible to show that $M_t$ has a unique solution for each $0\le t\le 1$.  In particular, $M_1=M$ has a solution, which determines the desired \slcf\ of $G'$.

\paragraph{Every free-collinear set is a free set:}

The proof that every free-collinear set is a free set is somewhat anticlimactic.  The following argument appears in \citet{bose.dujmovic.ea:untangling}, but similar ``perturb and scale'' arguments are fairly common.

Let $S:=(v_1,\ldots,v_s)$ be a free-collinear set in $G$ and let $x_1<\cdots<x_s$ and $y_1,\ldots,y_s$ be real numbers. Since $S$ is a free-collinear set in $G$, $G$ has a
 \slcf\ in $\Gamma_0$ where $v_i$ is at $(x_i,0)$ for each $i\in\{1,\ldots,s\}$.  It follows easily from the definition of \slcf\ that there exists some $\epsilon>0$ such that perturbing each of the vertices in $\Gamma_0$ by at most $\epsilon$ results in another \slcf\ $\Gamma_0'$.  Let $y:=\max\{|y_1|,\ldots,|y_s|\}$.  Define $\Gamma_0'$ by moving $v_i$ to $(x_i,\epsilon y_i/y)$, for each $i\in\{1,\ldots,s\}$ and leaving the other vertices of $G$ fixed.  Then each vertex in $\Gamma_0'$ has been moved a distance of at most $\epsilon$ from its location in $\Gamma_0$, so $\Gamma_0'$ is a \slcf\ of $G$.  Finally, multiply the \yy-coordinate of each vertex by $y/\epsilon$ to obtain a final straight-line drawing $\Gamma$.  This \yy-scaling is an affine transformation that does not introduce any crossings, so $\Gamma$ is a \slcf\ of $G$ and, for each $i\in\{1,\ldots,s\}$, $\Gamma$ places $v_i$ at $(x_i,y_i)$.

\section{Graph Classes with Large Free Sets}
\label{large_free_sets}
 What is the maximum size of a free set that is guaranteed to exist in any $n$-vertex planar graph?  \citet{DBLP:conf/wg/RavskyV11} observed that upper bounds on this value can be obtained from existing work on the circumference of cubic triconnected planar graphs. They observe that if a triangulation $G$ has a collinear set of size $\ell$ (and thus a large free set by \cref{equivalence}), then its dual graph $G^*$ (which is a cubic triconnected planar graph) has a cycle of length $\Omega(\ell)$. The length $c(G^*)$ of a longest cycle in $G^*$ is called the \defin{circumference} of $G^*$.  The circumference of cubic triconnected planar graphs has a long and rich history dating back to at least 1884 when Tait \cite{tait:remarks} conjectured that every such graph is Hamiltonian.\footnote{An $n$-vertex graph is \defin{Hamiltonian} if its circumference is equal to $n$.}

Tait's Conjecture was famously disproved in 1946 by Tutte, who gave an example of a non-Hamiltonian cubic triconnected planar graph having 46 vertices \cite{tutte:on}.  Repeatedly replacing vertices of Tutte's graph with copies of itself gives a family of graphs, $\langle G_i:i\in\mathbb{N}\rangle$ in which $G_i$ has $46\cdot 45^i$ vertices and circumference at most $45\cdot44^i$.  Stated another way, $n$-vertex members of the family have circumference $O(n^\sigma)$, for $\sigma=\log_{45}(44) < 0.9941$. The current best upper bound of this type is due to Gr\"unbaum and Walther \cite{GRUNBAUM1973364} who construct a family of cubic triconnected planar graphs in which $n$-vertex members have circumference $O(n^{\sigma})$ for $\sigma=\log_{23}(22)< 0.9859$.  The dual of such a graph is a triangulation having $\Theta(n)$ vertices whose largest free set has size $O(n^{\sigma})$.
\citet{cano.toth.ea:upper} use this upper bound on the circumference of cubic triconnected planar graphs, along with the Erd\H{o}s-Szekeres Theorem to show that $\fix_\mathcal{G}(n)\in O(n^{0.4948})$.

Since not all planar graphs have free sets of linear size, it is natural to ask which subclasses of planar graphs do.  Two obvious candidates are planar graphs of maximum degree $\Delta$ and planar graphs of treewidth at most $k$.  However, constructions like those described above can be used to rule out this possibility except for $\Delta < 7$ and $k<5$, as we now explain.

\citet{owens:regular} constructs a family of $n$-vertex cubic triconnected planar graphs whose faces have size at most $7$, and that contain no cycle of length $\Omega(n^{0.9976})$.  The dual of such a graph is a triangulation having $\Theta(n)$ vertices and maximum degree $7$ whose largest free set has size $O(n^{0.9976})$.

\citet{DBLP:conf/wg/RavskyV11} show that a construction based on the Barnette-Bos\'ak-Lederberg graph produces triangulations of treewidth at most $8$ whose largest free set has size $o(n)$.  \citet{dalozzo.dujmovic.ea:drawing} observe that the recursive construction based on Tutte's counterexample to Tait's Conjecture \cite{tutte:on} leads to a triangulation of treewidth at most $5$ whose largest free set has size $O(n^{0.9941})$.

Next, we begin by considering the subclasses of planar graphs that do admit linear-size free sets (in \cref{sec:linear}). We then study general planar graphs (in \cref{sec:planar}).  The relationship between the largest free sets and the circumference is then used to obtain strong bounds for bounded-degree planar graphs (in \cref{sec:degree}).

\subsection{Subclasses of planar graphs with linear-size free sets}\label{sec:linear}
A \defin{level planar drawing} of a graph is a  \slcf\  in the plane, such that the vertices are placed on a sequence of horizontal lines (called \defin{levels}), where each edge joins vertices in two consecutive levels.
For example, a natural \slcf\ of a tree $T$, with root $r$, places each vertex of $T$ that is at distance $i$ from $r$ on the line $y = -i$. Thus, trees have level planar drawings.

The following easy argument shows that $n$-vertex graphs that have level planar drawings have free sets of size at least $n/2$. See \cref{level_planar_fig}. Starting with a level planar drawing $\Gamma$, create a simple closed curve $C$ that visits alternate levels of $\Gamma$. Either the curve that visits the even levels or the curve that visits the odd levels contains at least $\lceil n/2\rceil$ vertices of $G$.  The curve $C$ does not intersect the interior of any edge of $\Gamma$.  Moreover, since $C$ does not contain any two consecutive levels of $\Gamma$, the set of vertices of $G$ that are on $C$ is an independent set of $G$, so $C$ intersects each edge of $\Gamma$ in at most one endpoint. Thus $C$ is a proper-good curve.

\begin{figure}
    \centering
    \begin{tabular}{cc}
       \includegraphics[page=2]{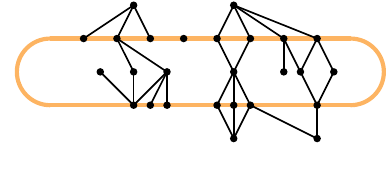} &
       \includegraphics[page=3]{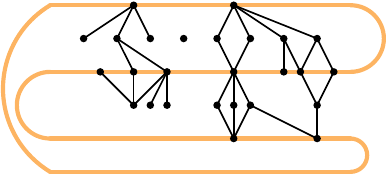}
    \end{tabular}
    \caption{Two proper-good curves in a level planar graph.}
    \label{level_planar_fig}
\end{figure}

Note that the preceding argument holds even if the definition of level planar drawings is relaxed to allow: consecutive vertices within a level to be adjacent; edges to go between non-consecutive levels, as long as they are not too far apart; and edges between different levels to be strictly $y$-monotone instead of straight. This leads to the following definition by \citet{bekos_et_al:LIPIcs.GD.2024.19}: An \defin{$s$-span weakly level planar} drawing is a crossing-free drawing in the plane, such that the vertices are placed on a sequence of parallel lines, where each edge $e$ is either a straight-line segment between two consecutive vertices on the same level (called a \defin{horizontal edge}) or a strictly $y$-monotone curve that intersects at most $s+1$ levels (called a \defin{vertical edge}). A graph is \defin{$s$-span weakly level planar} if it has an $s$-span weakly level planar drawing.  The following lemma formalizes the usefulness of this notion for obtaining free sets. This connection between $1$-span weakly level planar
level planar drawings have been observed by \cite{DBLP:conf/wg/RavskyV11}.

\begin{lem}\label{fs-weakly}
   Every $n$-vertex  $s$-span weakly level planar graph $G$ has a free set of size at least $\lceil n/(s+1)\rceil$.
\end{lem}
\begin{proof}
Refer to \cref{s_level_fig}.
Number the levels in an $s$-span weakly level planar drawing $\Gamma$ of $G$ by $0, 1, 2 \dots$. Then for some $i\in \{0, 1, \dots, s\}$, the union of levels $j = i\mod (s+1)$ has at least $\lceil n/(s+1)\rceil$ vertices, $S$ of $G$. Moreover, $G[S]$ is a forest of paths (induced by the horizontal edges and the vertices on these levels). As in the case of level planar graphs, construct a closed curve $C$ that contains these levels of $\Gamma$ (and thus the vertices of $S$) and such that $C$ intersects each vertical edge of $\Gamma$ at most once. This is possible since vertical edges are $y$-monotone and their endpoints lie on the levels whose difference (the absolute value) is at most $s$.  Such a curve $C$ is a proper-good curve and thus by \cref{equivalence}, $S$ is a free set whose vertices are ordered by their appearance on $C$.
\end{proof}

\begin{figure}
    \centering
    \includegraphics[page=2]{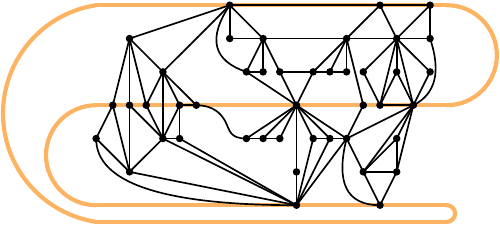}
    \caption{A $2$-span weakly level planar drawing and a proper-good curve that contains every third level.}
    \label{s_level_fig}
\end{figure}

Exactly the same proof shows the following ``hereditary'' variant of \cref{fs-weakly}:
\begin{lem}\label{fs-weakly-x}
   For any $s$-span weakly level planar graph $G$ and any subset $X$ of vertices in $G$, $G$ has a free set $S\subseteq X$ of size at least $\lceil |X|/(s+1)\rceil$.
\end{lem}

These lemmas have several immediate consequences. Firstly, as argued above, breath-first-search leveling of trees can be easily turned into level planar drawings (and thus  $1$-span weakly level planar drawings). Similarly,  breath-first-search leveling of outerplanar graphs can be turned into  $1$-span weakly level planar drawings of such graphs, as proved by \citet{JGAA-75}. The natural way to draw the $n \times n$ grid graph is a 1-level planar drawing. More generally, \citet{DBLP:journals/algorithmica/BannisterDDEW19} show that squaregraphs are $1$-span weakly level planar. A \defin{squaregraph} is a graph that has a crossing-free drawing in which each bounded face is a $4$-cycle and each vertex either belongs to the unbounded face or has four or more incident edges.
The same authors also show that Halin graphs are $1$-span weakly level planar. \citet{DBLP:journals/corr/abs-2311-14634} identified several classes of planar graphs that are  $s$-span weakly level  planar, for some constant $s$.

\begin{cor}\label{n_over_2}
    Let $G$ be a $n$-vertex tree, outerplanar graph, Halin graph, or square graph and let $X$ be any subset of the vertices of $G$.  The $G$ has a free set $S\subseteq X$ of size at least $\lceil |X|/2\rceil$.  In particular, $G$ has a free set of size at least $\lceil n/2\rceil$.
\end{cor}

The $\lceil n/2\rceil$ lower bound for trees appears in \cite{bose.dujmovic.ea:untangling} in the context of untangling.  The extension to outerplanar graphs appears in \cite{goaoc.kratochvil.ea:untangling}, also in the context of untangling.
For outerplanar graphs, a different construction of a proper-good curve was given by \citet{goaoc.kratochvil.ea:untangling}. It has a special structure, and the bound is slightly stronger. See \cref{pg_outerplanar}(c). In particular, the curve obtained there is contained in the closure of the outer face of $G$ and contains $G$ in the closure of its interior.  This structure turns out to be important when studying general planar graphs.  \citet{goaoc.kratochvil.ea:untangling} prove the following result, although it is not stated in this form. We include a variant of their proof.

\begin{figure}
    \centering
    \begin{tabular}{ccc}
        \includegraphics[page=1,trim={35 0 15 0}]{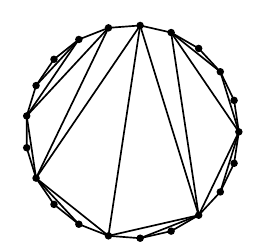} &
        \includegraphics[page=4]{figs/pgouterplanar} & 
        \includegraphics[page=2,trim={20 0 0 0}]{figs/pgouterplanar} \\
        (a) & (b) & (c)
    \end{tabular}
    \caption{An edge-maximal outerplane graph $G$; (b) a proper-good curve guaranteed by \cref{fs-weakly}/\cref{n_over_2}; and (c) a proper-good curve guaranteed by \cref{outerplanar_is}.}
    \label{pg_outerplanar}
\end{figure}

\begin{lem}\label{outerplanar_is}
    Let $G$ be an $n$-vertex biconnected outerplane graph for some $n\ge 4$. Then there exists a proper-good curve $C$ that is contained in the closure of the outer face of $G$, that contains $G$ in the closure of its interior, 
    and that passes through at least $\lceil n/2\rceil + 1$ vertices of $G$. 
    Thus the vertices in $C$ form a proper-good set of size at least $\lceil n/2\rceil+1$.
\end{lem}

Before proving this lemma, we emphasize that the condition that $C$ is contained in the closure of the outer face of $G$ does not imply the second condition—namely, that $G$ lies in the interior of $C$. These two properties will be important for identifying large free sets in general planar graphs. In particular, they ensure that $C$ can closely follow the outer face of an induced outerplane subgraph of a planar graph.

\begin{proof}[Proof of \cref{outerplanar_is}]
 Since $G$ is biconnected and $n\geq 4$, the boundary of the outer face is a cycle, denoted by $O$.  We may assume that $G$ is edge-maximal so that the boundary of each inner face of $G$ is a cycle of length $3$. 
 Consider a set $S\subseteq V(G)$ that has the following two properties: $G[S]$ induces a forest of paths and all the edges of $G[S]$ lie on $O$. It is easy to see that every such set $S$ is a proper-good set of $G$. Moreover, it is not hard to produce a proper-good curve that is contained in the closure of the outer face of $G$ and contains the vertices and edges of $G[S]$ by closely tracing the outer face of $G$.

We now prove that $G$ has such a set $S$ of claimed size. Let $T$ denote the weak dual of $G$. Then $T$ is a tree on at least two vertices since $n\geq 4$. Each vertex of $T$ has degree $1$, $2$ or $3$, so $T$ is a binary tree. Let $t_1$ and $t_3$ denote the number of degree $1$
 and degree $3$ vertices in $T$, respectively. It is well known that in every binary tree $t_1=t_3+2$.

Consider now the graph $G'$ obtained from $G$ by removing the edges on the outer face (in other words, $G'$ contains only the chords of $O$). Every independent set in $G'$ meets the two conditions imposed on $S$ earlier. Thus it remains to prove that $G'$ has an independent set $S$ of size at least $\frac{n}{2} + 1$. Construct $S$ greedily in $G'$ as follows: put in $S$ the vertex of $G'$ of minimum degree; remove that vertex and its neighbours from $G'$ to obtain a new $G'$; and repeat. $S$ is clearly an independent set in $G'$. It remains to show that it has the claimed size. In the moment a vertex is placed in $S$, its degree in the current $G'$ was 0, 1 or 2. Let $n_i$ denote the number of vertices that had degree $i$ when they were placed in $S$. Thus $|S|=n_0+n_1+n_2$. From the description of the algorithm, it follows that $n=n_0+2n_1+3n_2$. Finally, each leaf of  $T$ contributes $1$ to $n_0$. Thus $n_0\geq t_1$. Each vertex of $G'$ that contributes to $n_2$ corresponds to a unique face whose dual vertex has degree $3$ in $T$. Thus $n_2\leq t_3$. Combining that inequality with the earlier obtained equality, $t_1=t_3+2$, we get that $n_0\geq n_2 +2$. To summarize, we have the two equations and one inequality:
\[
    |S| =n_0+n_1+n_2 \qquad
    \qquad n =n_0+2n_1+3n_2 \qquad
    n_0- n_2  \geq 2 \enspace .
\]
Replacing $n_1$ in the first equality with $n_1=(n-n_0-3n_2)/2$ obtained from the second equality gives $|S|=\frac{n}{2}+\frac{n_0-n_2}{2}$. Combined with the last inequality, we obtain $|S|\ge n/2+1$, and the result follows from the fact that $|S|$ is an integer.
\end{proof}

\citet{DBLP:conf/wg/RavskyV11} extend this result even further to show that all graphs of treewidth $2$ have free-collinear sets of linear size. This extension does not follow from the argument used to prove \cref{fs-weakly}, since there is a class of planar graphs of treewidth at most $2$ whose graphs do not have $s$-span weakly level planar drawings for any fixed $s$ \cite{DBLP:journals/dcg/Biedl11}.

Having a weakly level planar drawing of small span is too strong a condition to provide a general tool for finding large free sets in wider classes of planar graphs.  In addition to the example of $2$-trees, there are planar graphs of bounded pathwidth, illustrated in \cref{biedl_fig} that do not have $o(n)$-span weakly level planar drawings (a result that can be derived from \citet[Theorem~5]{DBLP:journals/dcg/Biedl11}) and yet each graph in the class has an induced path of length $(n-1)/2$  that can be covered by a proper-good curve and thus the vertices on that path form a free set of size at least $(n-1)/2$, by \cref{equivalence}.

\begin{figure}
    \centering
    \begin{tabular}{cc}
        \includegraphics[page=1]{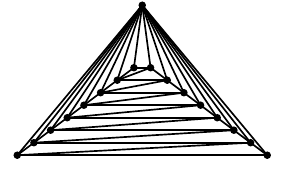} &
        \includegraphics[page=2]{figs/biedl}
    \end{tabular}
    \caption{A graph of pathwidth-$2$ with no $o(n)$-span weakly level planar drawing that contains a free set of size $(n-1)/2$}
    \label{biedl_fig}
\end{figure}

\citet{dalozzo.dujmovic.ea:drawing} show that the linear bound for treewidth $2$ graphs \cite{DBLP:conf/wg/RavskyV11} extends to planar graphs of treewidth at most $3$:

\begin{thm}[\cite{dalozzo.dujmovic.ea:drawing}]\label{fs-tw3}
Every $n$-vertex planar graph of treewidth at most three has a free set of size at least $\lceil\frac{n-3}{8}\rceil$.
\end{thm}

As noted in the introduction to this section, \cref{fs-tw3} cannot be generalized to all $n$-vertex planar graphs of treewidth at most $5$. This leaves open the question of whether a linear bound is possible for planar graphs of treewidth at most $4$.

For planar graphs of \emph{large} treewidth, \citet{dalozzo.dujmovic.ea:drawing} use the fact that any planar graph of treewidth $k$ contains a $k\times k$ grid-minor to show that planar graphs with large treewidth have large free sets:

\begin{thm}[\cite{dalozzo.dujmovic.ea:drawing}]\label{fs-tw}
Every planar graph of treewidth $k$ has a free set of size $\Omega(k^2)$.
\end{thm}

\Cref{fs-tw} implies that all $n$-vertex planar graphs of treewidth $\Omega(\sqrt{n})$ have a free set of linear size (a vast generalization of the above observation that square grids have linear size free sets).

Rather than considering planar graphs of small treewidth, one can also consider planar graphs of small (maximum) degree.  \citet{dalozzo.dujmovic.ea:drawing} prove the following result in this vein:

\begin{thm}[\cite{dalozzo.dujmovic.ea:drawing}]\label{fs-cubic}
Every $n$-vertex planar triconnected cubic graph has a free set of size at least $\lceil\frac{n}{4}\rceil$.
\end{thm}

\citet{dalozzo.dujmovic.ea:drawing} suggest the possibility of extending \cref{fs-cubic} to show the existence of a linear-sized free set in any planar graph of maximum degree $3$.  As discussed in the introduction to this section, no such result is possible for all planar graphs of maximum degree $7$.  For $\Delta\in\{3,4,5,6\}$ it is still open whether a linear bound is possible for all planar graphs of maximum degree $\Delta$.

\subsection{Free Sets in Planar Graphs}\label{sec:planar}

\begin{thm}
   \label{root_n}
   Every $n$-vertex planar graph $G$ has a free set of size at least $\sqrt{n/2}$.
\end{thm}

A version of \cref{root_n} with an $\Omega(\sqrt{n})$ bound is due to \citet{bose.dujmovic.ea:untangling} though the authors at the time were working on the untangling problem discussed in the introduction, so their result is never stated in terms of free sets.  They prove that $\fix_\mathcal{G}(n)=\Omega(n^{1/4})$ by proving that a triangulation $G$ either contains an induced outerplane graph of size $\Omega(\sqrt{n})$ or a free-collinear set of size $\Omega(\sqrt{n})$.  At the time, the equivalence between proper-good sets, collinear sets, and free-collinear sets was not known, so their proof includes both combinatorial and geometric elements (including a proof that free-collinear and free sets are equivalent). In the following, we extract these combinatorial elements from \cite{bose.dujmovic.ea:untangling} and use \cref{outerplanar_is} (proven in \cite{goaoc.kratochvil.ea:untangling}) to give a self-contained proof of \cref{root_n} with the best currently known bound of $\sqrt{n/2}$.

The \defin{interior}, $\interior(G)$, of a near-triangulation $G$ is the interior of the cycle that bounds the outer face of $G$.

\begin{lem}\label{dumb_lemma}
  Let $G$ be a near-triangulation, and let $v$ and $w$ be two points (possibly vertices) on the boundary of the outer face of $G$ such that no edge of the outer face of $G$ contains $v$ and $w$. Then there exists a simple curve $C:[0,1]\to\interior(G)\cup\{v,w\}$ with endpoints $v$ and $w$ that is proper with respect to $G$.
\end{lem}

\begin{proof}
  If $vw$ is an (internal) edge of $E(G)$, then the curve $C$ consists of the edge $vw$, which clearly satisfies the requirements of the lemma. 
   \begin{figure}
      \centering
      \includegraphics[page=3]{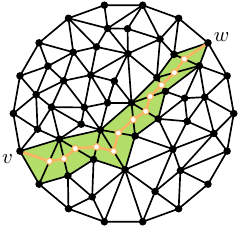}
      \caption{Constructing a proper curve $C$ from $v$ to $w$ through $\interior(G)$.}
      \label{dumb_lemma_fig}
  \end{figure}

 Otherwise, since the edge $vw$ is not an edge of $G$ and since $G$ is a near-triangulation, there is no internal face of $G$ that contains both $v$ and $w$ on its boundary. Refer to \cref{dumb_lemma_fig}.  
 Let $G^+$ be the weak dual of $G$.  Let $P$ be a shortest path in $G^+$ having one endpoint in a face that contains $v$ and the other endpoint in a face that contains $w$.  Let $e^+_1,\ldots,e^+_t$ be the sequence of edges in $P$.  For each $i\in\{1,\ldots,t\}$, the edge $e_i^+$ in $G^+$ corresponds to some inner edge $e_i$ of $G$. For each $i\in\{1,\ldots,t\}$, let $p_i$ be any point in the interior of $e_i$, let $p_0:=v$, and let $p_{t+1}=w$.  Then, for each $i\in\{1,\ldots,t+1\}$ there is an inner face $f_i$ of $G$ with the points $p_{i-1}$ and $p_i$ on the boundary of $f_i$.  Construct an open simple curve $C$ that visits $p_0,\ldots,p_{t+1}$ in order in such a way that the portion of $C$ between $p_{i-1}$ and $p_i$ is contained in the interior of $f_i$, for each $i\in\{1,\ldots,t+1\}$. 

 Clearly $C$ is contained in $\interior(G)\cup\{v,w\}$, so all that remains is to show that $C$ is proper with respect to $G$.  By construction, $C$ intersects each of the edges $e_1,\ldots,e_t$, $C$ intersects the edges of $G$ incident to $v$, $C$ intersects the edges of $G$ incident to $w$, and $C$ avoids all other edges of $G$.  Since $P$ is a shortest path, none of $e_1,\ldots,e_t$ is incident to $v$ or $w$. Since $P$ is a path, $C$ intersects each of $e_1,\ldots,e_t$ in a single point. Finally, $C$ intersects each of the edges of $G$ incident to $v$ (respectively $w$) in a single point, namely $v$ (respectively, $w$).  Thus $C$ is proper with respect to $G$.
\end{proof}

\begin{proof}[Proof of \cref{root_n}]
If $G_0$ is a triangulation such that $G$ is a spanning subgraph of $G_0$, then clearly, any free set in $G_0$ is also a free set in $G$. Thus we may assume, without loss of generality, that $G$ is a triangulation.
 The steps in this proof are illustrated in \cref{canonical}. Fix a \embedding\ of $G$ and let $v_1$, $v_2$, and $v_n$ be the three vertices on the outer face of $G$.

    \begin{figure}
        \centering
        \begin{tabular}{cc}
        \includegraphics[page=12]{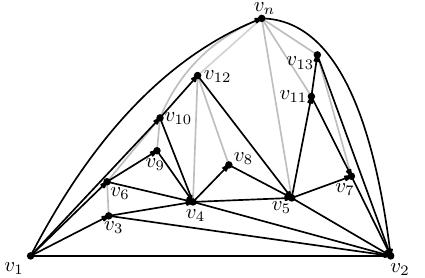}  &
        \includegraphics[page=20]{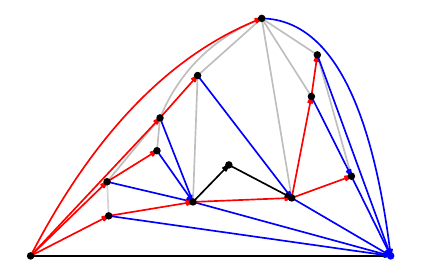} \\
        (a) & (b) \\
        \includegraphics[page=13]{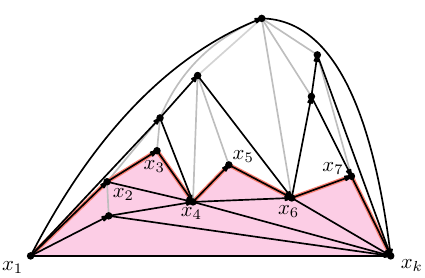} &
        \includegraphics[page=15]{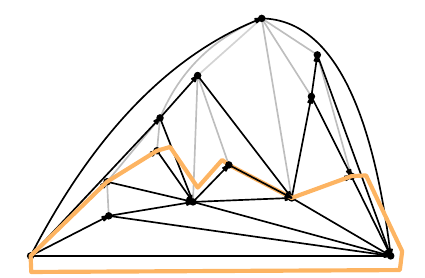} \\
        (c) & (d) \\
        \includegraphics[page=17]{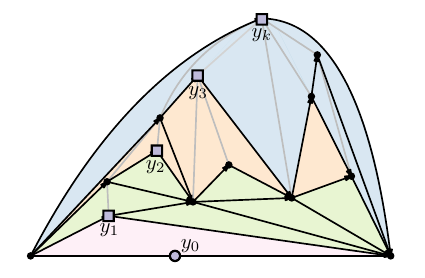} &
         \includegraphics[page=18]{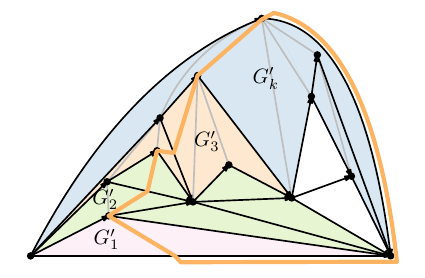} \\
         (e) & (f)
        \end{tabular}
        \caption{(a) and (b) A canonical ordering $v_1,\ldots,v_n$ of a triangulation $G$ and the resulting frame $F$.\newline
        (c) and (d): a chain of $(V(G),\prec)$ and the resulting proper-good curve $C$ \newline
        (e) and (f): An antichain of $(V(G),\prec)$ and the resulting proper-good curve $C$.}
        \label{canonical}
    \end{figure}

    We use a \defin{canonical ordering} $v_1,\ldots,v_n$ of $V(G)$, which has the following property: For each $i\in\{3,\ldots,n\}$,  the induced graph $G_i:=G[\{v_1,\ldots,v_i\}]$ is a near-triangulation that contains the vertex $v_i$ and the edge $v_1v_2$ on its outer face.  The existence of such an ordering is proven by \citet{defraysseix.pach.ea:how}.  In this ordering, the introduction of $v_i$ introduces $d_i\ge 2$ new edges $v_iw_{i,1},\ldots,v_iw_{i,d_i}$ to $G_i$ that do not appear in $G_{i-1}$.  The edges $e_{i}:=v_{i}w_{i,1}$ and $e'_i:=v_{i}w_{i,d_i}$ appear on the outer face of $G_i$ and the vertices $w_{i,2},\ldots,w_{i,d_i-1}$ on the outer face of $G_{i-1}$ no longer appear on the outer face of $G_i$, as illustrated in \cref{canonical}(a).
    This defines a \defin{frame} $F$ where $V(F):=V(G)$ and $E(F):=\{v_1v_2\}\cup\bigcup_{j=3}^n \{e_i,e'_i\}$.  We treat $F$ as a directed acyclic graph where $v_1v_2$ is directed towards $v_2$ and $w_{i,1},v_i,w_{i,d_i}$ is a directed path.  In this way, $F$ has a single source, $v_1$ and a single sink, $v_2$.\footnote{For readers familiar with Schnyder woods \cite{schnyder:embedding}, $F$ can be obtained by taking the union of two trees $T_1$ and $T_2$ (rooted at $v_1$ and $v_2$, respectively) in a Schnyder Wood, directing each edge of $T_1$ away from its root and directing each edge of $T_2$ towards its root, as illustrated in \cref{canonical}(b).}  Since the final frame $F$ is a directed acyclic graph, its transitive closure defines a partially ordered set $(V(G),\prec)$ in which $v\prec w$ if and only if $F$ contains a directed path from $v$ to $w$.

    First, consider some maximal chain $x_1\prec\cdots\prec x_k$ in this partial order.  Since $F$ is maximal, $x_ix_{i+1}$ is an edge of $F$ for each $i\in\{1,\ldots,k-1\}$. Since $F$ has a single source $v_1$ and a single sink $v_2$, $x_1=v_1$ and $x_k=v_2$.  
    Thus $C_x:=x_1,\ldots,x_k$ is a cycle in $G$, as illustrated in \cref{canonical}(c).
    Suppose that $G$ contains some edge $x_ix_j$ that is a chord of $C_x$.    We now argue that the edge $x_ix_j$ is embedded in the interior of the cycle $C_x$.  Without loss of generality, suppose $x_i$ appears after $x_j$ in the canonical order, so $x_i=v_a$ and $x_j=v_b$ for some $a> b$. Then the edge $x_ix_j=v_av_b$ is in the graph $G_a=G[\{v_1,\ldots,v_a\}]$, so $v_b\in\{w_{a,1},\ldots,w_{a,d_a}\}$. The two neighbours $x_{i-1}$ and $x_{i+1}$ of $x_i$ in $C_x$ are not in the set $w_{a,2},\ldots,w_{a,d_{a}-1}$, so the interior of the cycle $C_x$ contains the interior of the cycle $v_a,w_{a,1},\ldots,w_{a,d_a}$. Therefore, the edge $v_av_b=x_ix_j$ is in the interior of $C_x$. (Indeed, the only possibility is that $v_b\in\{w_{a,1},w_{a,d_a}\}$ )

     Therefore, $G_x:=G[\{x_1,\ldots,x_k\}]$ is an induced outerplane subgraph of $G$ whose outer face is bounded by the cycle $C_x$. Each edge of $G$ with two endpoints in $G_x$ is contained in the closure of the interior of $C_x$. Thus each edge that is not in $E(G_x)$ intersects the exterior of $C_x$ and has at most one endpoint on $C_x$. Let $E_0\subseteq E(G)\setminus E(G_x)$ be the set of edges with no endpoint on $C_x$, and $E_1\subseteq E(G)\setminus E(G_x)$ be the set of edges with exactly one endpoint on $C_x$.  By \cref{outerplanar_is}, there is simple closed curve $C$ that contains a set $S$ of at least $k/2$ vertices of $G_x$, that is contained in the closure of the outer face of $G_x$, that contains $G_x$ in the closure of its interior, and that is proper and good with respect to $G_x$. These properties ensure that $C$ visits the vertices of $G_x$ in the same order as they appear on $C_x$ and that $C$ can be made to follow closely $C_x$  -- in fact, so closely that $C$ intersects each edge of $E_1$ in one point and does not intersect any edge in $E_0$. Since the remaining edges of $G$ (those in $E(G_x)$) are contained in the closure of the interior of $C_x$, the curve $C$ intersects each edge in $E(G)$ in at most one point,  as illustrated in \cref{canonical}(d).  Thus $C$ is a proper-good curve for $G$ that contains the vertices in $S$.  By \cref{equivalence}, $S$ is a free set in $G$ of size at least $k/2$.

    Next  consider some maximal antichain $S:=y_1,\ldots,y_k:=v_{i_1},\ldots,v_{i_k}$ of $(V(G),\prec)$ of length $k>1$ ordered by canonical ordering so that $i_1<\cdots<i_k$, as illustrated in  \cref{canonical}(e).  Since $k>1$, the set $S$ does not contain $v_1$ or $v_2$.  Since $S$ is maximal and does not contain $v_1$ or $v_2$, we have $y_k=v_n$.  Consider the sequence of cycles $C_1,\ldots,C_k$ where $C_j$ is the cycle that bounds the outer face of $G_{i_j}$.  By definition of  canonical ordering, the interior of $C_j$ contains the interior of $C_{j-1}$ for each $j\in\{2,\ldots,k\}$.  The nesting of these cycles is illustrated in \cref{canonical}(e), where each new colour shows the interior of $C_j$ that is not contained in the interior of $C_{j-1}$. Furthermore, the interior of $C_{j}$ must contain the vertex $y_{j-1}$ since, otherwise both $y_{j-1}$ and $y_j$ are on the outer face of $G_{i_j}$, which would mean that $y_{j-1}$ and $y_j$ are comparable. Consider now the union of two cycles $C_{j-1}$ and $C_{j}$ (or rather the union of the two closed curves that represent these two cycles in the \embedding\ of $G$). 
    This union has one (bounded) face $f$ that contains both $y_{j-1}$ and $y_{j}$ on its boundary. The boundary of $f$ is a cycle $D_j$ in $G$ comprised of two paths: a path in $C_{j-1}$ containing  $y_{j-1}$ and a path in $C_j$ containing  $y_j$.  Note that the union of all the cycles $D_1, \dots D_k$,
    (or rather the union of the $k$ closed curves that represent these $k$ cycles in the \embedding\ of $G$) defines $k+1$ faces in the plane.
    These faces are illustrated in \cref{canonical}(f), where the interior of each face is assigned its own colour.

    We now construct a proper-good curve $C$ for $G$ that contains the vertices of $S$.  Let $y_0$ be a point in the interior of the edge $v_1v_2$. Let $G_1':=G_{i_1}$ and, for each $j\in\{2,\ldots,k\}$, let $G_j'$ be the near-triangulation whose outer face is bounded by the cycle $D_j$ described in the previous paragraph, and whose inner faces are faces of $G$. Graphs $G_j'$ are illustrated in \cref{canonical}(f) by having their interiors shaded. Observe that $G'_1,\ldots,G'_k$ have pairwise disjoint interiors.  By \cref{dumb_lemma}, there is a simple curve $I_j$:$[0,1]\to\interior(G'_{j})$ with endpoints $y_{j-1}$ and $y_j$ that is proper with respect to $G'_{j}$, for each $j\in\{1,\ldots,k\}$.  Since $G'_1,\ldots,G'_k$ have pairwise disjoint interiors, the curve obtained by taking the union of $I_1,\ldots,I_k$ is simple and it is proper with respect to $G$.  We complete this curve into a cycle by connecting $y_k$ to $y_0$ through the outer face of $G$.
    By \cref{equivalence}, $S$ is a free set of size $k$.

    To complete the proof, we use Dilworth's Theorem, which guarantees that the poset $(V(G),\prec)$ contains a chain of size at least $\sqrt{2n}$ or an antichain of size at least  $\sqrt{n/2}$. In either case, we obtain a free set of size at least $\sqrt{n/2}$.
\end{proof}

The following generalization of \cref{root_n}, observed by \citet{dujmovic:utility}, has an almost identical proof, except that one considers the induced poset $(X,\prec)$ rather than $(V(G),\prec)$:

\begin{thm}
   \label{root_x}
   For every planar graph $G$ and every $X\subseteq V(G)$, $X$ contains a free set of size at least $\Omega(\sqrt{|X|})$.
 \end{thm}

\subsection{Free Sets in Max-Degree-$\Delta$ Planar Graphs}\label{sec:degree}

\citet{dalozzo.dujmovic.ea:drawing} suggest the possibility that, since upper bounds on the circumference of dual graphs can be used to obtain upper bounds on the size of free sets, maybe
lower bounds on circumference can be used to prove the existence of large collinear sets.  \citet{dujmovic.morin:dual} show that, for planar graphs of bounded degree, this is indeed the case.
In short, they show that a triangulation $G$ of maximum degree $\Delta$ whose dual has circumference $c(G^*)$ has a free set of size $\Omega(c(G^*)/\Delta^4)$. A series of results has steadily improved the lower bounds on the circumference of $n$-vertex  (not necessarily planar)  cubic triconnected graphs \cite{barnette:trees,bondy.simonovits:longest,jackson:longest,bilinksi.jackson.ea:circumference,liu.yu.zhang:circumference}.
 The current record is held by \citet{liu.yu.zhang:circumference} who show that, for any $n$-vertex cubic triconnected graph $G^*$, $c(G^*)\in\Omega(n^{0.8})$.

In the remainder of this subsection, we describe some of the techniques used in \cite{dujmovic.morin:dual} to establish the $\Omega(c(G^*)/\Delta^4)$ result.   Cycles in $G^*$ are relevant to free sets because every cycle in $G^*$ corresponds to a proper-good curve in $G$.  The resulting curve $C$ does not contain any vertices of $G$, but it is natural to try reroute $C$ to obtain a new curve that goes through some vertices of $G$. \cref{fig:move} (a) depicts a situation where this rerouting fails because the resulting curve is no longer proper because it intersects the red edge in \cref{fig:move}(b) in two points. This leads to the following definition that lays out conditions under which such rerouting of curve $C$ is safe.  We say that a vertex $v$ of $G$ is \defin{caressed} by $C$ if the edges of $G$ that are incident to $v$ and intersected by $C$ appear consecutively around $v$, as depicted in \cref{fig:move}(c) (the thick edges crossed by $C$ are consecutive around $v$ unlike in \cref{fig:move}(a) where the set of edges crossed by $c$ is not consecutive around $v$).  If $C$ caresses $v$ and intersects the edges $vv_1,\ldots,vv_{r}$ (see \cref{fig:move}(c)) then there is a sequence of faces $f_0,\ldots,f_{r}$ where $vv_i$ is the edge shared between $f_{i-1}$ and $f_i$. Then the portion of $C$ that intersects $f_0,\ldots,f_r$  can be replaced with a curve $C'$ that enters $f_0$, proceeds directly to $v$, and immediately exits $f_r$ (see \cref{fig:move}(d)).  Because the original curve $C$ does not intersect any other edges incident to $v$, the modified curve $C'$ is also a proper-good curve.  This operation can be repeated on any set of caressed vertices that form an independent set:

\begin{figure}
    \centering
    \begin{tabular}{cc}
        \includegraphics[page=3]{figs/cycle-to-curve} &
        \includegraphics[page=4]{figs/cycle-to-curve} \\
        (a) & (b) \\
        \includegraphics[page=1]{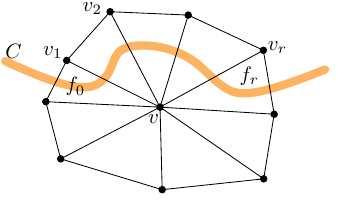} &
        \includegraphics[page=2]{figs/cycle-to-curve} \\
        (c) & (d) \\
    \end{tabular}
    \caption{Attempting to reroute a dual cycle (a proper-good curve $C$) into a new proper-good curve $C'$ so that it contains $v$.} \label{fig:move}
\end{figure}

\begin{lem}\label{reroute}
  Let $C$ be a proper-good curve in a triangulation $G$ and let $S$ be an independent set of at least two vertices in $G$ that are each caressed by $C$.  Then $S$ is a free set in $G$.
\end{lem}

\begin{proof}
   For each $v\in S$, reroute $C$ as described above so that $C$ contains $v$.  Since $v$ is caressed by $C$, the rerouting that takes place at $v\in S$ causes $C$ to intersect each edge incident to $v$ in exactly one point, namely $v$, but does not change the intersection of $C$ with any edges not incident to $v$.  Since the vertices in $S$ form an independent set, all of these rerouting operations do not cause $C$ to intersect any edge of $G$ in more than one point. Thus $C$ is a proper-good curve that contains $S$ so, by \cref{equivalence} $S$ is a free set in $G$. (The condition that $S$ has at least two vertices avoids the case in which $C$ crosses all the edges incident to a single vertex $v$.)
\end{proof}

The requirement that a curve $C$ caresses a vertex $v$ of $G$ is equivalent to requiring that the intersection of $C$ with the face $f_v$ of $G^*$ that contains $v$ is a path. When this happens, we say that $C$ \defin{caresses} $f_v$. Thus finding a large collinear set in $G$ is equivalent to finding a cycle in $G^*$ that caresses many faces of $G^*$.

\begin{lem}
  Let $G$ be a triangulation and let $G^*$ be the dual of $G$.  If some cycle $C$ in $G^*$ caresses $k$ faces of $G^*$, then $G$ contains a free set of size at least $k/4$.
\end{lem}

\begin{proof}
  By the Four Colour Theorem \cite{robertson.seymour.ea:four-colour}, the faces of $G^*$ can be coloured with four colours so that no two faces that share an edge are assigned the same colour.  One of the resulting colour classes contains at least $k/4$ faces of $G^*$ that are caressed by $C$.  These faces correspond to an independent set $S$ of vertices of $G$ that are caressed by $C$.  The lemma now follows from \cref{reroute}.
\end{proof}

Unfortunately, there exist cubic triconnected planar graphs $G^*$ with no large faces that have a Hamiltonian cycle that caresses only four faces of $G^*$. One such graph is shown in \cref{two_caressed}.  The main technical result of \citet{dujmovic.morin:dual} is to show that such cycles can be modified to produce cycles that caress many faces:

\begin{figure}
    \centering
    \includegraphics{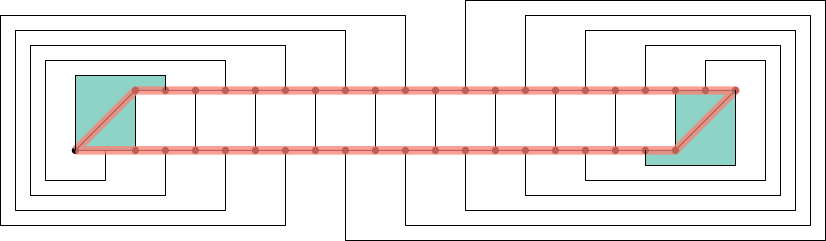}
    \caption{A Hamiltonian cycle in a cubic triconnected planar graph that caresses only four faces.}
    \label{two_caressed}
\end{figure}

\begin{lem}[\cite{dujmovic.morin:dual}]\label{dual_caressed}
    If $G^*$ is an $n$-vertex cubic triconnected planar graph with no faces of size greater than $\Delta$, then there exists a cycle $C$ in $G^*$ that caresses at least $\Omega(c(G^*)/\Delta^4)$ faces of $G^*$.
\end{lem}

\begin{proof}[Proof Sketch]
Refer to \cref{surgery_fig}.
The proof of \cref{dual_caressed} is far too long to include here in any detail, so we give a high-level sketch.  In the following informal sketch, we treat $\Delta$ as a fixed constant.  Let $C$ be a cycle in $G^*$. Say that $C$ \defin{touches} a face $f$ of $G^*$ if $C$ and $f$ share at least one edge. Since each face of $G^*$ has at most $\Delta$ edges and each edge of $C$ touches two faces of $G^*$ (one inside $C$ and one outside of $C$), the number of faces of $G^*$ that are touched by $C$ is at least $2|C|/\Delta$.  At least $|C|/\Delta$ of these faces are in the interior of $C$ and at least $|C|/\Delta$ of these are in the exterior of $C$.

The proof defines a subgraph $H$ of $G^*$ that includes all the edges of $C$ and such that any cycle in the dual graph $H^*$ of $H$ contains faces inside and outside of $C$.  Removing the edges from $H^*$ that correspond to edges of $C$ produces two trees $T_0$ and $T_1$,  where $T_0$ contains faces of $H$ in the interior of $C$ and $T_1$ contains faces of $H$ in the exterior of $C$.  For each $b\in\{0,1\}$, the tree $T_b$ has two important properties:
\begin{compactenum}
    \item Each leaf of $T_b$ corresponds to a face $f$ of $H$ that contains at least one face of $G^*$ that is caressed by $C$.
    \item Let $f\in V(T_b)$ be a face of $H$ that has degree $\delta$ in $T_b$ and that contains $\tau$ faces of $G^*$ touched by $C$, $\kappa$ of which are caressed by $C$.  Then $3\kappa+2\delta \ge \tau$.
\end{compactenum}
The second of these properties says that any node of $T_b$ that contains many faces touched by $C$ either caresses many of these faces or has a high degree in $T_b$.  The number of leaves in $T_b$ is $2+\sum_{x}(\deg_{T_b}(x)-2)$, where the sum is over all non-leaf nodes $x$ of $T_b$.  Combining this with the first of these properties implies that $C$ caresses $\Omega(|C|)$ faces of $G^*$ or that $T_b$ has $\Omega(|C|)$ nodes but at most $\epsilon |C|$ leaves, for some small $\epsilon >0$.  In the former case, we are done. In the latter case, we conclude that the vast majority of nodes in $T_b$ have degree two.

\begin{figure}
    \centering
    \begin{tabular}{cc}
      \includegraphics[page=1]{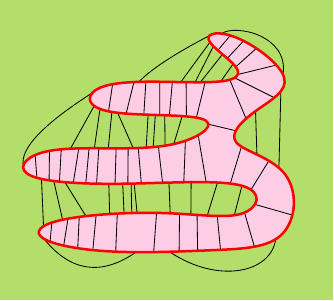} &
      \includegraphics[page=2]{figs/surgery.pdf}
    \end{tabular}
    \caption{Performing surgery on the trees $T_0$ (pink) and $T_1$ (green) in order to increase the number of leaves.}
    \label{surgery_fig}
\end{figure}
Since $T_0$ and $T_1$ share the boundary $C$, this can be used to show that $T_0$ contains a path $P_0$ of degree-$2$ nodes and $T_1$ contains a path $P_1$ of degree-$2$ nodes, each of size bounded by a constant (depending on $\Delta$) such that moving the nodes in $P_0$ from $T_0$ to $T_1$ and moving the nodes in $P_1$ from $T_1$ to $T_0$ results in two new subgraphs $T_0''$ and $T_1''$ of $H^*$ that contain at least one more leaf than $T_0$ and $T_1$.  The graphs $T_0''$ and $T_1''$ are not necessarily trees, but the boundary they share defines a new cycle $C'$ from which we can define two trees $T_0'$ and $T_1'$ as before.  These two new trees differ from $T_0$ and $T_1$ only in a small neighbourhood of the paths $P_0$ and $P_1$, which ensures that $T_0'$ and $T_1'$ also have at least one more leaf than $T_0$ and $T_1$. At this point, the procedure is repeated on $C'$.  Each application of this procedure may decrease the length of the cycle of $C$ by at most a constant (for a fixed $\Delta$), so this procedure can be repeated $\Omega(|C|)$ times, which produces a cycle $\tilde{C}$ that defines two trees $\tilde{T}_0$ and $\tilde{T}_1$ having a total of $\Omega(|C|)$ leaves, each of which contains a face of $G^*$ caressed by $C$.
\end{proof}

If $G^*$ is the dual of a triangulation $G$ then the size of each face of $G^*$ is the degree of the corresponding vertex in $G$.  This gives the following consequence of \cref{dual_caressed}.

\begin{cor}\label{dual_triangulation_circumference}
  Let $G$ be an $n$-vertex triangulation of maximum degree $\Delta$ and let $G^*$ be the dual of $G$.  Then $G$ has a free set of size $\Omega(c(G^*)/\Delta^4)$.
\end{cor}

\cref{dual_triangulation_circumference} shows that, for triangulations of maximum degree $\Delta=O(1)$,  the dual circumference and the size of the largest free set differ by only a constant factor.  

A result of \citet{kant.bodlaender:triangulating} shows that one can add edges to any planar graph $G_0$ of maximum degree $\Delta$ to obtain a triangulation $G$ of maximum degree at most $\lceil 3\Delta/2\rceil+11$. Combining this with the best known lower bounds on the circumference of cubic triconnected graphs \cite{liu.yu.zhang:circumference} and \cref{dual_triangulation_circumference} gives:

\begin{cor}\label{dual_planar_concrete}
  Let $G$ be an $n$-vertex planar graph of maximum degree $\Delta\ge 1$.  Then $G$ has a free set of size $\Omega(n^{0.8}/\Delta^4)$.
\end{cor}

\section{Applications}
\label{applications}

In this section, we discuss applications of free sets.  For each of these ($s$) applications, and each of the ($t$) classes of planar graphs known to have large free sets, the use of free sets gives a result.  Rather than providing an exhaustive list of $s\times t$ corollaries, we state the relationship between free sets and the application and give the most general corollary, the one about planar graphs obtained from \cref{root_n}.

\subsection{Untangling}

As discussed in the introduction, the original application of free sets was untangling. By \cref{free_to_fix}, each lower bound (on the size of the largest free set in a graph family) obtained in \cref{sec:linear,sec:planar,sec:degree}, readily turns into the square root of the same bound for the untangling problem on that family. For example, \cref{root_n,free_to_fix} yield:

\begin{cor}\label{untangling}
    Let $\mathcal G$ be the class of all planar graphs. For every $n$-vertex graph $G\in {\mathcal G}$,  $\fix(G)\ge (n/2)^{1/4}$.  In other words, $\fix_{\mathcal{G}}(n)\ge (n/2)^{1/4}$. 
\end{cor}

We do not know if this bound is asymptotically optimal. \citet{cano.toth.ea:upper} proved that for the class of planar graphs $\mathcal G$, $\fix_{\mathcal{G}}(n)\in O(n^{0.4965})$.

In contrast, applying the known bounds on the largest free sets coupled with \cref{free_to_fix} gives tight bounds for some graph families. Specifically, the linear bounds on free sets in \cref{sec:linear} and \cref{free_to_fix} imply an $\Omega(n^{1/2})$ untangling bound for these families. \citet{bose.dujmovic.ea:untangling} showed that there is a class of trees, namely forests of stars, ${\mathcal S}$, such that  $\fix_{\mathcal{S}}\in O(n^{1/2})$. Thus the $\Omega(n^{1/2})$  untangling bound is tight for all families in \cref{sec:linear},  except for the class of cubic triconnected planar graphs. The stars in ${\mathcal T}$ have unbounded degree. Thus a tight bound on untangling cubic triconnected planar graphs is unknown.

\subsection{Universal point subsets}

A set of points $P$ is \defin{universal} for a set of planar graphs if every graph from the set has a \slcf\ where all of its vertices map to distinct points in $P$. Arguably, the most famous open problem in graph drawing, attributed to Bojan Mohar (1988), asks if there exists a constant $c$ such that, for every $n$ there exists a pointset of size $c\cdot n$ that is universal for the class of all $n$-vertex planar graphs.  It is known that for all large enough $n$,  no universal pointset of size $n$ exists for the class of $n$-vertex planar graphs -- as first proved by \citet{defraysseix.pach.ea:how}. The same authors also proved that the $O(n) \times O(n) $ integer grid is universal for all $n$-vertex planar graphs and thus universal pointsets of size $O(n^2)$ exist. Currently the best known lower bound  on the size of a smallest universal pointset for $n$-vertex planar graphs
is $1.239n-o(n)$ \cite{scheucher.schrezenmaier.ea:note}
and the best known upper bound is $n^2/4 - O(n)$ \cite{bannister.cheng.ea:superpatterns}. Closing the gap between $\Omega(n)$ and $O(n^2)$ is a major, and likely very difficult, graph drawing problem, open since 1988 \cite{defraysseix.pach.ea:small,defraysseix.pach.ea:how}.

Various related notions have been introduced and studied in the literature with the aim of better understanding the universal pointset problem.  A set $P$ of $k$ points in the plane is a \defin{universal point subset} for a family of planar graphs, if for every graph $G$ in the family, there exists a \slcf\ in which $k$ vertices of $G$ are placed at the $k$ points in $P$.  Universal point subsets were introduced by \citet{DBLP:conf/isaac/AngeliniBEHLMMO12}. They proved that a particular convex chain of $\lceil\sqrt{n}\rceil$ points is a universal point subset for the class of all $n$-vertex planar graphs. \citet{DBLP:journals/corr/abs-1212-0804} then proved that every set of $\Omega(\log n)$ points, as well as, every set of $\lceil n^{1/3}\rceil$ points in convex position is a universal point subset for the class of $n$-vertex planar graphs.

The following lemma and its consequences, first observed by \citet{dujmovic:utility}, is immediate from the definition of free sets.

\begin{lem}\label{lem:ups}
    Let $k$ be a positive integer and $\mathcal{F}$ a family of planar graphs such that every graph in the family has a free set of size at least $k$. Then every set of $k$ points in the plane is a universal point subset for $\mathcal{F}$.
\end{lem}

Note that this lemma does not just prove that there exists a set of $k$ points that is a universal point subset for $\mathcal{G}$, it proves that \emph{every} set of $k$ points is a universal point subset for $\mathcal{G}$. This distinction turns out to be useful in other applications.  Combining \cref{lem:ups,root_n} gives:

\begin{cor}
    Every set of $\lceil \sqrt{n/2}\rceil$ points is a universal point subset for the class of $n$-vertex planar graphs.
\end{cor}

It has been known for a long time that every set of $n$ points in general position is a universal point (sub)set for the family of $n$-vertex outerplanar graphs \cite{GMPP,DBLP:conf/gd/Bose97,DBLP:conf/cccg/CastanedaU96}.
\cref{lem:ups} implies that all the families studied in \cref{sec:linear}, also admit universal point subsets of linear size. This includes, for example, planar graphs of treewidth at most $3$  by \cref{fs-tw3}, a strict superclass of the class of outerplanar graphs.

\subsection{Simultaneous Geometric Embeddings}\label{sim}

Simultaneous geometric embeddings were introduced by
\citet{DBLP:journals/comgeo/BrassCDEEIKLM07}. Since then, there has been a plethora of work on the subject on many variants of the problem-- see, for example, 
a survey by \citet{BKR-HGD}. Common variants of the problem include those in which the mapping between the vertices of the two graphs is given and those in which the mapping is not given.

\subsubsection{Without Mapping}

A sequence of graphs $G_1, G_2, \dots, G_r$  are said to have a \defin{simultaneous geometric embedding without mapping (SGE-nomap)} if there exists a pointset $P$ of size $\max\{|V(G_1)|,\ldots,|V(G_r)|\}$ such that
each of $G_1,\ldots,G_r$ has a \slcf\ where all of its vertices are mapped to distinct points in $P$. The following well known and still wildly open problem was asked by \citet{DBLP:journals/comgeo/BrassCDEEIKLM07} in 2003, and is also listed among selected list of graph drawing problems in \cite[Problem~12]{DBLP:conf/gd/BrandenburgEGKLM03}: Does every pair of $n$-vertex planar graphs have a SGE-nomap (for every positive integer $n$)? The statement is known \emph{not} to be true when ``pair'' (that is, $r=2$) is replaced by a bigger constant \cite{DBLP:journals/jgaa/CardinalHK15, scheucher.schrezenmaier.ea:note}, currently being $r=30$  \cite{DBLP:conf/gd/Steiner23}. More generally, \citet{DBLP:conf/gd/Steiner23} showed that for every large enough $n$ there exists a sequence of at most $(3 + o(1)) \log_2 (n)$ $n$-vertex planar graphs that do not have a SGE-nomap.

While the problem of \citet{DBLP:journals/comgeo/BrassCDEEIKLM07} still seems to be out of reach, in a different direction \citet{angelini.evans.ea:sefe} write: ``What is the largest
$k \leq n$ such that every $n$-vertex planar graph and every
$k$-vertex planar graph admit a geometric simultaneous embedding without mapping? Surprisingly, we are not aware of any super-constant lower bound for the value of $k$?'' The corollary below, noticed first by \citet{dujmovic:utility}, answered their question with the help of free sets, as detailed in the next lemma.

\begin{lem}\label{lem:sim}
Let $G_1$ be a planar graph with a free set of size at least $k$. Let $G_2$ be a planar graph on at most $k$ vertices. Then $G_1$ and $G_2$ admit SGE-nomap. 
\end{lem}
\begin{proof}
By F\'ary's theorem, $G_2$ has a straight-line crossing-free drawing on some set, $P_2$, of $|V(G_2)|\leq k$ points. Since $G_1$ has a free set $S$ of size at least $k$, $G_1$ has a \slcf\ where $|P_2|$  vertices of $G_1$ are mapped to distinct points in $P_2$. Consider now the set of points, $P$, defined by the vertices in the drawing of $G_1$. This set is our desired pointset as it is a set of $|V(G_1)|$ points such that each of $G_1$ and $G_2$ has a \slcf\ where all of its vertices are mapped to distinct points of  $P$.
\end{proof}

\Cref{lem:sim,root_n} immediately imply the following, aforementioned result.

\begin{cor}\label{fab}
For every $n$ and every $k\leq \sqrt{n/2}$, every $n$-vertex planar graph and every $k$-vertex planar graph admit a SGE-nomap.
\end{cor}

One can combine the many lower bounds (on the size of the largest free set in a graph family) obtained in \cref{sec:linear,sec:planar,sec:degree} to obtain results of similar flavour to \cref{fab}. One example is the following: for every $n$ and every $k\leq n/4$, every $n$-vertex planar graph and every $k$-vertex triconnected cubic planar graph admits a SGE-nomap.

Note that universal point subsets of size $k$ do not imply results on SGE-nomap, akin to \cref{fab}.  In particular, imagine that, as in the proof of \cref{lem:sim}, one starts with a straight-line crossing-free drawing of the smaller graph $G_2$. Even if the larger graph $G_1$ is in a class of graphs that have universal point subsets of size $k$, there is no guarantee that $G_1$ can be drawn on the specific pointset of size $k$ used by the drawing of $G_2$. To be able to use universal point subsets for SGE-nomap problems, one needs the stronger variant of ``all pointsets of size $k$'' to be subset universal---something that is guaranteed by free sets of size $k$.

\subsubsection{With Mapping}

Two $n$-vertex planar graphs $G_1$ and $G_2$ on \emph{the same vertex set}, $V := V (G_1) = V (G_2)$, are said to have a \defin{$k$-partial simultaneous geometric embedding with mapping} (\defin{$k$-PSGE-withmap})  if there exists a set $V':=\{v_1,\ldots,v_k\}\subseteq V$, and a set $P:=\{p_1,\ldots,p_k\}$ of points such that each of $G_1$ and $G_2$ has a \slcf\ in which $v_i$ is mapped to $p_i$, for each $i\in\{1,\ldots,k\}$.
  PSGE-withmap on the whole vertex set (i.e. $n$-PSGE-withmap) has been widely studied, leading to mostly negative results (thus giving another motivation to introduce this partial version).  For example, it is known that, for every large enough $n$, there are pairs of $n$-vertex planar graphs that do not have $n$-PSGE-withmap \cite{DBLP:journals/comgeo/BrassCDEEIKLM07}. In fact the same is true for graphs from very simple families of planar graphs, for example: for an $n$-vertex tree and an $n$-vertex path \cite{DBLP:journals/jgaa/AngeliniGKN12}, for an $n$-vertex planar graph and an $n$-vertex matching \cite{DBLP:journals/jgaa/AngeliniGKN12} and for three $n$-vertex paths \cite{DBLP:journals/comgeo/BrassCDEEIKLM07}.

The $k$-PSGE-withmap problem was introduced\footnote{In \cite{DBLP:conf/gd/EvansKSS14}, they use the abbreviation $k$-PSGE for what we call $k$-PSGE-withmap. Also, the definition of $k$-PSGE-withmap in  \cite{DBLP:conf/gd/EvansKSS14} has one additional requirement, which states that if $v, w\in V$ are mapped to the same point in $D_i$ and $D_j$,  then $v=w$.
However, this additional requirement can always be met by the fact that it is possible to perturb any subset of vertices in a \slcf\ without introducing crossings.  (This fact is used, for example, in the proof that every free-collinear set is free.)}
by \citet{DBLP:conf/gd/EvansKSS14}
who proved that any two $n$-vertex trees have an $11n/17$-PSGE-withmap.  Their proof uses column planarity, which is the topic of the next section. \citet{barba.hoffmann.ea:column} proved that any two $n$-vertex outerplanar graphs have an $n/4$-PSGE-withmap.
\citet{DBLP:conf/gd/EvansKSS14}  also observed  that \cref{untangling}, the untangling result,  implies that every pair of $n$-vertex planar graphs has an $\Omega(n^{1/4})$-PSGE-withmap.  Namely, start with a \slcf\ of $G_1$. Since $G_1$ and $G_2$ have the same vertex set, the drawing of $G_1$ (or rather the positions of its vertices in the plane) defines a straight-line drawing of $G_2$ (that almost certainly has crossings). Untangling $G_2$ while keeping $\Omega(n^{1/4})$ of its vertices fixed (which is possible by \cref{untangling}) gives the result.

\begin{thm}[\cite{DBLP:conf/gd/EvansKSS14}]\label{psge}
Every pair of $n$-vertex planar graphs has an $\Omega(n^{1/4})$-PSGE-withmap. 
\end{thm}

However, the above untangling argument fails if we try to apply it one more time. The following generalization of the $k$-PSGE problem to more than two graphs illustrates this: Given any set  $Q:=\{G_1,\ldots,G_r\}$ of planar graphs on \emph{the same vertex set}, $V$, we say that the graphs in $Q$ have a \defin{$k$-partial simultaneous geometric embedding with mapping} ($k$-PSGE-withmap) if there exists a set $V':=\{v_1,\ldots,v_k\}\subseteq V$, and a set $P:=\{p_1,\ldots,p_k\}$ of points such that each graph in $Q$ has a \slcf\ in which $v_i$ is mapped to $p_i$, for each $i\in\{1,\ldots,k\}$.

If we try to mimic the earlier untangling argument that proves \cref{psge}, it fails for $r=3$ already since we need to be able to guarantee that when $G_3$ is untangled the set of its vertices that stays fixed has a large intersection with the set that remained fixed when untangling $G_2$. It is here that we need the stronger version of \cref{root_n}, namely \cref{root_x}.

We start by presenting a lemma about PSGE-withmap for two graphs via free sets. Its proof uses the common trick of taking advantage of the fact that points in the plane have two degrees of freedom, and thus one ordering can be imposed on the \xx-coordinates and a different ordering on \yy-coordinates (see for example \cite{DBLP:journals/jgaa/BarbaEHKSS17}).

\begin{lem}\label{free-PSGE-two}
Let $G_1$ and $G_2$ be two planar graphs on the same vertex set $V:=V(G_1)=V(G_2)$. Let $S\subseteq V$ be an (unordered) free set in $G_1$ and an (unordered) free set in $G_2$, then $G_1$ and $G_2$ have a $|S|$-PSGE-withmap.
\end{lem}

\begin{proof}
Let $S:=\{v_1,\ldots,v_s\}$, let $S_1$ be a permutation of $S$ such that $S_1$ is an ordered free set in $G_1$, and let $S_2$ be permutation of $S$ such that $S_2$ is an ordered free set in $G_2$.  Define the pointset $P:=\{p_1,\ldots,p_s\}$ in which $p_i$ has \xx-coordinate equal to the position of $v_i$ in $S_1$ and \yy-coordinate equal to the position of $v_i$ in $S_2$, for each $i\in\{1,\ldots,s\}$.
 By the definition of free set, $G_1$ has \slcf\ drawing with $v_i$ at $p_i$, for each $i\in\{1,\ldots,s\}$.  By the definition of free set---with the roles of \xx-coordinates and \yy-coordinates reversed---$G_2$ has a \slcf\ with $v_i$ at $p_i$, for each $i\in\{1,\ldots,s\}$.  Thus $G_1$ and $G_2$ have a $|S|$-PSGE-withmap.
\end{proof}

\cref{root_x,free-PSGE-two} gives another proof of \cref{psge}.  By \cref{root_x} (applied to $G_1$ and $V$), there exists $S_1\subseteq V$ of size at least $\sqrt{n/2}$ that is a free set in $G_1$.  By \cref{root_x} (applied to $G_2$ and $S_1$), there exists $S\subseteq S_1$ of size at least $\sqrt{|S_1|/2}\ge n^{1/4}/2^{3/4}$.  Thus $S$ is a free set in $G_1$ and in $G_2$.  By \cref{free-PSGE-two}, $G_1$ and $G_2$ have a $\lceil n^{1/4}/2^{3/4}\rceil$-PSGE-withmap.  In exactly the same way, \cref{fs-weakly-x,free-PSGE-two} give a proof that any two $n$-vertex graphs, each of which is one of the following: a tree, an outerplanar graph, Halin graph, or a square graph,  have a $(n/4)$-PSGE-withmap.

\begin{lem}\label{free-PSGE-many}
Let $Q:=\{G_1,\ldots,G_r\}$ be a set of $r\geq 2$ planar graphs on the same vertex set $V$. Let $S\subseteq V$ be an unordered free set in $G_i$, for each $i\in\{1,\ldots,r\}$. Then the graphs in $Q$ have a $|S|^{1/2^{(r-2)}}$-PSGE-withmap.
\end{lem}

\begin{proof}
   For each $i\in\{2,\ldots,r\}$, let $S_i$ be a permutation of $S$ that is an ordered free set in $G_i$.  Suppose that, for some $i\in\{3,\ldots,r\}$, $S_{i-1}'$ is a subsequence of $S_2$ that defines an (ordered) free set in each of $G_2,\ldots,G_{i-1}$.
   Since $S_i$ and its reversal are both free sets in $G_i$, Dilworth's Theorem implies that there is a subsequence $S_i'$ of $S_{i-1}'$ of size at least $|S_{i-1}'|^{1/2}$ that is a free set in $G_i$ and is therefore a free set in each of $G_1,\ldots,G_i$.  By starting with $S_2':=S_2$, it follows that there exists an ordered set $S':=S_{r}'$ of size $|S|^{1/2^{r-2}}$ that is a free set in each of $G_2,\ldots,G_r$.

   Now apply \cref{free-PSGE-two} to the graphs $G_1$ and $G_2$ with the set $S$ to obtain \slcf\ of both graphs in which each vertex in $S$ (and therefore, the vertices in $S'$) appears at the same location in both drawings. In these drawings, the \yy-coordinates of the vertices in $S'$ are increasing.  By the definition of free set (exchanging the roles of \xx- and \yy-coordinates), each of $G_3,\ldots,G_k$ has a \slcf\  in which each vertex in $S'$ has the same location in each drawing.
\end{proof}

Starting with $X:=V$ and using $r$ applications of \cref{root_x} yields a set $S$ of size at least $n^{1/2^r}/2$ that is a free set in each of $G_1,\ldots,G_r$.  Applying \cref{free-PSGE-many} to $S$ and $G_1,\ldots,G_r$ gives the following corollary:

\begin{cor}\label{free-PGSE-n}
  Let $Q:=\{G_1,\ldots,G_r\}$ be a set of $r\geq 2$ planar graphs on the same vertex set $V$. Then the graphs in $Q$ have a $(n^{1/4^{r-1}}/2)$-PSGE-withmap. 
\end{cor}

If the graphs in $Q$ each come from one of the classes covered by \cref{n_over_2}, then repeated applications of \cref{n_over_2} gives a set $S$ of size at least $n/2^r$ that is free in each of $G_1,\ldots,G_r$, in which case the bound in \cref{free-PGSE-n} becomes $n^{1/2^{r-2}}/4$.

\subsection{Column planarity}

Given a planar graph $G$, a set $R\subseteq V(G)$ is  \defin{column planar} in $G$ if the vertices of $R$ can be assigned distinct \xx-coordinates such that for an assignment of \yy-coordinates to the vertices in $R$ such that the resulting $|R|$ points have no three on a line,  there exists a \slcf\ of $G$ in which each vertex $v\in R$ is placed to its assigned coordinates. $R$ is  \defin{strongly column planar} if the ``no three on a line condition'' is removed from the definition.  The following lemma is immediate from these definition and \cref{equivalence}:

\begin{lem}\label{col-eq}
    A (sub)set of vertices of a planar graph is a strongly column planar set if and only if it is an unordered free set. 
\end{lem}

 Being strongly column planar implies being column planar, thus all the bounds from  \cref{sec:linear,sec:planar,sec:degree}  apply to column planarity.  For example, \cref{root_n} implies that every $n$-vertex planar graph has a column planar set of size $\Omega(\sqrt{n})$. \cref{dual_planar_concrete} implies that every $n$-vertex bounded-degree planar graph has a column planar set of size $\Omega(n^{0.8})$; and so on.

Column planar sets were introduced by \citet{DBLP:conf/gd/EvansKSS14} motivated by applications to partial simultaneous geometric embeddings. Notions similar to column planarity were studied by \citet{DBLP:journals/comgeo/Estrella-BalderramaFK09} and \citet{DBLP:journals/jgaa/GiacomoDKLS09}.
\citet{DBLP:journals/jgaa/BarbaEHKSS17} proved that $n$-vertex trees have column planar sets of size $14n/17$. Note that that does not imply that trees have free sets of that size, since being column planar does not imply being strongly column planar. In particular, if $G$ is a 3-cycle, then $V(G)$ is a column planar set for any assignment of \xx-coordinates to $V(G)$. However, no \xx-coordinate assignment can be turned into a \slcf\ with all the vertices having the same \yy-coordinate. Thus $V(G)$ is not strongly column planar.

\section{A One-Bend Variant}
\label{one_bend_section}

A \defin{$k$-bend \embedding} of a planar graph $G$ is a \embedding\ of $G$ in which each edge is represented by a polygonal chain consisting of at most $k+1$ line segments.  Thus a $0$-bend \embedding\ is a \slcf.  This leads naturally to a definition of a $k$-bend free set introduced by \citet{bose.dujmovic.ea:connected}.  An ordered set $S:=(v_1,\ldots,v_s)$ of vertices in a planar graph $G$ is a \defin{$k$-bend free set} if, for any $x_1<\cdots<x_s$ and any $y_1,\ldots,y_s$, there exists a $k$-bend \embedding\ $\Gamma$ of $G$ such that $(x_i,y_i)$ is the location of $v_i$ in $\Gamma$, for each $i\in\{1,\ldots,s\}$.  As an application of \cref{cds}, below, \citet{bose.dujmovic.ea:connected} prove the following result:

\begin{thm}[\cite{bose.dujmovic.ea:connected}]\label{one_bend_thm}
    Every $n$-vertex planar graph has a $1$-bend free set of size at least $11n/21$.
\end{thm}

They prove \cref{one_bend_thm} by proving the following result, which is the main result in \cite{bose.dujmovic.ea:connected}:

\begin{thm}[\cite{bose.dujmovic.ea:connected}]\label{cds}
    Every $n$-vertex triangulation has a spanning tree with at least $11n/21$ leaves.
\end{thm}

\cref{cds} improves a longstanding bound of $n/2$ that has at least two different proofs \cite{albertson.berman.ea:graphs,angelini.evans.ea:sefe}. The proof of \cref{cds} is well outside the scope of this survey. Instead, we prove the following lemma that, along with \cref{cds}, immediately implies \cref{one_bend_thm}.

\begin{lem}\label{tree_to_obfs}
    Let $G$ be a planar graph and let $T$ be a spanning tree of $G$.  Then the leaves of $T$ are a one-bend collinear set in $G$.
\end{lem}

\begin{proof}
    Refer to \cref{one_bend_reduction}.
    Fix some \slcf\ of $G$, which also fixes a \slcf\ of $T$. Fatten the drawing of $T$ by some arbitrarily small value $\epsilon>0$ by taking the Minkowsky sum of $T$ with a disc of radius $\epsilon$. The boundary of this fattened tree is a simple closed curve $C_0$. Construct a curve $C$ by deforming $C_0$ in an $\epsilon$-neighbourhood of each leaf $v$ of $T$ so that it contains $v$.   If $\epsilon$ is sufficiently small, then $C$ is a simple closed curve that intersects each edge $vw$ of $G$ in at most $2$ points: one point within distance $\epsilon$ of $v$ and one point within distance $\epsilon$ of $w$.  Let $S:=(v_1,\ldots,v_s)$ be the leaves of $T$ in the order they are encountered while traversing $C$.  Subdivide each edge of $G$ by placing a vertex in the center of the edge, and call the resulting graph $G'$.  Then $C$ is a proper-good curve for $G'$ that contains $S$.  Therefore, $S$ is a free set in $G'$.  Therefore, for any $x_1<\cdots<x_s$ and any $y_1,\ldots,y_s$, $G'$ has a \slcf\ in which $v_i$ is placed at $(x_i,y_i)$ for each $i\in\{1,\ldots,s\}$.  Any such drawing of $G'$ gives a $1$-bend \embedding\ of $G$ in which $v_i$ is placed at $(x_i,y_i)$ for each $i\in\{1,\ldots,s\}$.  Thus $S$ is a $1$-bend collinear set in $G$.
\end{proof}

\begin{figure}
    \centering
    \includegraphics[page=7]{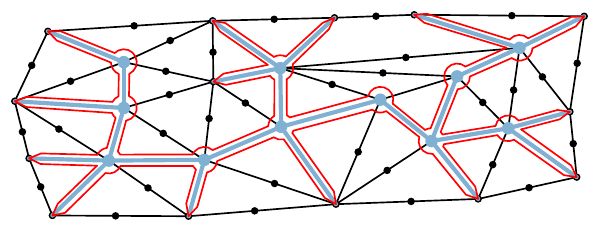}
    \caption{The leaves of a spanning tree of $G$ are a $1$-bend collinear set of $G$.}
    \label{one_bend_reduction}
\end{figure}

The obvious question is whether the bound in \cref{one_bend_thm} is of the correct form. Perhaps planar graphs have $1$-bend collinear sets of size $n-o(n)$.  We can rule out this possibility with the following construction.

\begin{thm}
    For infinitely many values of $n$, there exists an $n$-vertex planar graph that has no $1$-bend collinear set of size greater than $10n/11$.
\end{thm}

\begin{proof}
    The Goldner–Harary graph \cite{goldner1975note}, $G_0$ is an $11$-vertex triangulation that is not Hamiltonian. We claim that any $1$-bend free set in $G_0$ has size at most $10$. Suppose, for the sake of contradiction, that $S:=V(G_0)$ is a $1$-bend free set in $G_0$.  Then $G_0$ has a $1$-bend \embedding\ $D$ with all vertices of $S$ placed on the \xx-axis.  Since $D$ is a $1$-bend \embedding\ and all vertices of $G_0$ are on the \xx-axis, no edge properly crosses the \xx-axis. 

We now show that the existence of such a \embedding\ $D$ implies that $G_0$ is Hamiltonian, contradicting our initial assumption. Let $v$ and $w$ be any two consecutive vertices of $G_0$ along the \xx-axis. Suppose, for contradiction, that $vw \not\in E(G_0)$. Since a line segment between $v$ and $w$ does not cross any edge of $D$, adding the edge $vw$ to $G_0$ yields a planar graph with more edges than $G_0$. But this contradicts the fact that $G_0$ is a triangulation and therefore edge-maximal. Hence, every pair of consecutive vertices along the \xx-axis must be connected by an edge in $G_0$, implying that $G_0$ contains a Hamiltonian path. Similarly, a curve can be added to $D$ connecting the first and the last vertex on the \xx-axis without crossing any existing edge. An argument equivalent to the one above shows that this edge must also exist in $G_0$. It follows that $G_0$ contains a Hamiltonian cycle, completing the contradiction.

    For any positive integer $k$, let $G$ be a graph consisting of $k$ disjoint copies of $G_0$.  Then $G$ has $n=11k$ vertices. Any $1$-bend free set in $G$ must exclude one vertex from each copy of $G_0$ and therefore has size at most $10k=10n/11$.
\end{proof}

\section{Open Problems}
\label{conclusion}

We conclude with a list of open problems:

\begin{enumerate}
    \item What is the largest value $\alpha$ such that every $n$-vertex planar graph has a free set of size $\Omega(n^{\alpha})$?  Currently, we know that $1/2\le \alpha \le 0.9859$.  Does $\alpha=\beta$, where $\beta$ is the largest value such that every $n$-vertex cubic triconnected planar graph contains a cycle of length $\Omega(n^{\beta})$?

    \item Does every $n$-vertex planar graph of maximum degree $3$, $4$, $5$, or $6$ contain a free set of size $\Omega(n)$?

    \item Does every $n$-vertex planar graph of treewidth $4$ contain a free set of size $\Omega(n)$?

    \item What is the largest value of $\beta$ such that every $n$-vertex planar graph has a $1$-bend free set of size $\beta n-o(n)$?

    \item Say that a subset $S$ of vertices in a planar graph $G$ is \defin{mostly free} if, for every set $P$ of $|S|$ points \emph{in general position}, there exists a \slcf\ of $G$ in which each vertex in $S$ is drawn at a distinct point in $P$.  Is there a constant $\gamma >0$ such that every mostly free set $S$ contains a free subset of size at least $\gamma|S|$?  The class of outerplanar graphs shows that $\gamma \le 2/3$ since any set of $n$ points in general position is universal for the class of $n$-vertex outerplanar graphs \cite{GMPP,DBLP:conf/gd/Bose97,DBLP:conf/cccg/CastanedaU96}. Thus the entire vertex set of every outerplanar graph is mostly free, but there exists $n$-vertex outerplanar graphs with no free set of size greater than $\lceil n/2\rceil$.

    \item \cref{root_x,n_over_2} allow us to choose any set $X$ of vertices in $G$ and find a large subset of $X$ that is a free set of $G$, and this turns out to be useful in a number of applications.  \cref{dual_planar_concrete,fs-cubic}, which currently give the strongest lower bounds for planar graphs of bounded degree and triconnected cubic planar graphs, do not have this flexibility. 
    \begin{enumerate}
        \item Is the following strengthening of \cref{dual_planar_concrete} true: For any planar graph $G$ of bounded degree and any subset $X$ of vertices of $G$, there exists a set $S\subseteq X$ of size $\Omega(|X|^{0.8})$ that is a free set of $G$? 
        
        \item Is the following strengthening of \cref{fs-cubic} true: For any triconnected cubic planar graph $G$ and any subset $X$ of vertices of $G$, there exists a set $S\subseteq X$ of size $\Omega(|X|)$ that is a free set of $G$?
    \end{enumerate}
\end{enumerate}

In some of the applications discussed in \cref{applications}, free sets are a convenient tool but may be more powerful than necessary, and better bounds may be possible using more direct methods.  Here are two examples:

\begin{enumerate}\setcounter{enumi}{6}
    \item Do universal point subset of size $\Omega (n^{1/2+\epsilon})$ for some $\epsilon>0$ exist for the class of $n$-vertex planar graphs. Currently, there is nothing that rules out a bound of $\Omega(n)$?
    
    \item Although free sets and an application of Dilworth's Theorem currently give the best bounds for untangling planar graphs and several other graph classes, better bounds may be possible using more direct techniques.  The $\Omega(n^{2/3})$ bound on $\fix_{\mathcal{C}}(G)$ for untangling cycles \cite{cibulka:untangling} gives one example in which free sets are not the best approach.
\end{enumerate}

\subsection*{Acknowledgements}

We thank G\"unter Rote and the anonymous referees for detailed and insightful comments, which significantly improved the quality of this manuscript.

\bibliographystyle{plainurlnat}
\bibliography{main}

\begin{thebibliography}{60}
\providecommand{\natexlab}[1]{#1}
\providecommand{\url}[1]{\texttt{#1}}
\providecommand{\urlprefix}{URL }
\expandafter\ifx\csname urlstyle\endcsname\relax
  \providecommand{\doi}[1]{\href{https://dx.doi.org/#1}{\nolinkurl{doi:#1}}}\else
  \providecommand{\doi}[1]{\href{https://dx.doi.org/#1}{\nolinkurl{doi:#1}}}\fi
\providecommand{\eprint}[2][]{\href{https://dx.doi.org/10.48550/arXiv.#2}{\nolinkurl{doi:10.48550/arXiv.#2}}}

\bibitem[{Albertson et~al.(1990)Albertson, Berman, Hutchinson, and
  Thomassen}]{albertson.berman.ea:graphs}
Michael~O. Albertson, David~M. Berman, Joan~P. Hutchinson, and Carsten
  Thomassen.
\newblock Graphs with homeomorphically irreducible spanning trees.
\newblock \emph{J. Graph Theory}, 14(2):247--258, 1990.
\newblock \doi{10.1002/JGT.3190140212}.

\bibitem[{Angelini et~al.(2012{\natexlab{a}})Angelini, Binucci, Evans, Hurtado,
  Liotta, Mchedlidze, Meijer, and Okamoto}]{DBLP:conf/isaac/AngeliniBEHLMMO12}
Patrizio Angelini, Carla Binucci, William~S. Evans, Ferran Hurtado, Giuseppe
  Liotta, Tamara Mchedlidze, Henk Meijer, and Yoshio Okamoto.
\newblock Universal point subsets for planar graphs.
\newblock In \emph{Algorithms and Computation - 23rd International Symposium,
  {ISAAC} 2012}, volume 7676 of \emph{LNCS}, pages 423--432. Springer,
  2012{\natexlab{a}}.
\newblock \doi{10.1007/978-3-642-35261-4\_45}.

\bibitem[{Angelini et~al.(2016)Angelini, Evans, Frati, and
  Gudmundsson}]{angelini.evans.ea:sefe}
Patrizio Angelini, William~S. Evans, Fabrizio Frati, and Joachim Gudmundsson.
\newblock {SEFE} without mapping via large induced outerplane graphs in plane
  graphs.
\newblock \emph{J. Graph Theory}, 82(1):45--64, 2016.
\newblock \doi{10.1002/JGT.21884}.

\bibitem[{Angelini et~al.(2012{\natexlab{b}})Angelini, Geyer, Kaufmann, and
  Neuwirth}]{DBLP:journals/jgaa/AngeliniGKN12}
Patrizio Angelini, Markus Geyer, Michael Kaufmann, and Daniel Neuwirth.
\newblock On a tree and a path with no geometric simultaneous embedding.
\newblock \emph{J. Graph Algorithms Appl.}, 16(1):37--83, 2012{\natexlab{b}}.
\newblock \doi{10.7155/JGAA.00250}.

\bibitem[{Bannister et~al.(2014)Bannister, Cheng, Devanny, and
  Eppstein}]{bannister.cheng.ea:superpatterns}
Michael~J. Bannister, Zhanpeng Cheng, William~E. Devanny, and David Eppstein.
\newblock Superpatterns and universal point sets.
\newblock \emph{J. Graph Algorithms Appl.}, 18(2):177--209, 2014.
\newblock \doi{10.7155/JGAA.00318}.

\bibitem[{Bannister et~al.(2019)Bannister, Devanny, Dujmovi{\'c}, Eppstein, and
  Wood}]{DBLP:journals/algorithmica/BannisterDDEW19}
Michael~J. Bannister, William~E. Devanny, Vida Dujmovi{\'c}, David Eppstein,
  and David~R. Wood.
\newblock Track layouts, layered path decompositions, and leveled planarity.
\newblock \emph{Algorithmica}, 81(4):1561--1583, 2019.
\newblock \doi{10.1007/S00453-018-0487-5}.

\bibitem[{Barba et~al.(2017)Barba, Evans, Hoffmann, Kusters, Saumell, and
  Speckmann}]{DBLP:journals/jgaa/BarbaEHKSS17}
Luís Barba, William Evans, Michael Hoffmann, Vincent Kusters, Maria Saumell,
  and Bettina Speckmann.
\newblock Column planarity and partially-simultaneous geometric embedding.
\newblock \emph{J. Graph Algorithms Appl.}, 21(6):983--1002, 2017.
\newblock \doi{10.7155/JGAA.00446}.

\bibitem[{Barba et~al.(2015)Barba, Hoffmann, and
  Kusters}]{barba.hoffmann.ea:column}
Luís Barba, Michael Hoffmann, and Vincent Kusters.
\newblock Column planarity and partial simultaneous geometric embedding for
  outerplanar graphs.
\newblock In \emph{Abstracts of the 31st European Workshop on Computational
  Geometry (EuroCG)}, pages 53--56. 2015.

\bibitem[{Barnette(1966)}]{barnette:trees}
David Barnette.
\newblock Trees in polyhedral graphs.
\newblock \emph{Canadian Journal of Mathematics}, 18:731--736, 1966.

\bibitem[{Bekos et~al.(2024)Bekos, Da~Lozzo, Frati, Gupta, Kindermann, Liotta,
  Rutter, and Tollis}]{bekos_et_al:LIPIcs.GD.2024.19}
Michael~A. Bekos, Giordano Da~Lozzo, Fabrizio Frati, Siddharth Gupta, Philipp
  Kindermann, Giuseppe Liotta, Ignaz Rutter, and Ioannis~G. Tollis.
\newblock {Weakly Leveled Planarity with Bounded Span}.
\newblock In Stefan Felsner and Karsten Klein, editors, \emph{32nd
  International Symposium on Graph Drawing and Network Visualization (GD
  2024)}, volume 320 of \emph{Leibniz International Proceedings in Informatics
  (LIPIcs)}, pages 19:1--19:19. Schloss Dagstuhl -- Leibniz-Zentrum f{\"u}r
  Informatik, Dagstuhl, Germany, 2024.
\newblock \doi{10.4230/LIPIcs.GD.2024.19}.

\bibitem[{Biedl(2011)}]{DBLP:journals/dcg/Biedl11}
Therese Biedl.
\newblock Small drawings of outerplanar graphs, series-parallel graphs, and
  other planar graphs.
\newblock \emph{Discret. Comput. Geom.}, 45(1):141--160, 2011.
\newblock \doi{10.1007/S00454-010-9310-Z}.

\bibitem[{Bilinski et~al.(2011)Bilinski, Jackson, Ma, and
  Yu}]{bilinksi.jackson.ea:circumference}
Mark Bilinski, Bill Jackson, Jie Ma, and Xingxing Yu.
\newblock Circumference of 3-connected claw-free graphs and large eulerian
  subgraphs of 3-edge-connected graphs.
\newblock \emph{J. Comb. Theory {B}}, 101(4):214--236, 2011.
\newblock \doi{10.1016/J.JCTB.2011.02.009}.

\bibitem[{Bl{\"a}sius et~al.(2013)Bl{\"a}sius, Kobourov, and Rutter}]{BKR-HGD}
Thomas Bl{\"a}sius, Stephen~G. Kobourov, and Ignaz Rutter.
\newblock Simultaneous embedding of planar graphs.
\newblock In Roberto Tamassia, editor, \emph{Handbook of Graph Drawing and
  Visualization}, Discrete Mathematics and Its Applications, pages 349--381.
  CRC Press, 2013.

\bibitem[{Bondy and Simonovits(1980)}]{bondy.simonovits:longest}
J.~Adrian Bondy and Mikl{\'o}s Simonovits.
\newblock Longest cycles in 3-connected 3-regular graphs.
\newblock \emph{Canadian Journal of Mathematics}, 32:987--992, 1980.

\bibitem[{Bose(1997)}]{DBLP:conf/gd/Bose97}
Prosenjit Bose.
\newblock On embedding an outer-planar graph in a point set.
\newblock In \emph{Graph Drawing, 5th International Symposium, {GD} '97},
  volume 1353 of \emph{LNCS}, pages 25--36. Springer, 1997.
\newblock \doi{10.1007/3-540-63938-1\_47}.

\bibitem[{Bose et~al.(2023)Bose, Dujmovi{\'c}, Houdrouge, Morin, and
  Odak}]{bose.dujmovic.ea:connected}
Prosenjit Bose, Vida Dujmovi{\'c}, Hussein Houdrouge, Pat Morin, and Saeed
  Odak.
\newblock Connected dominating sets in triangulations.
\newblock \emph{CoRR}, abs/2312.03399, 2023.
\newblock \eprint{2312.03399}.

\bibitem[{Bose et~al.(2009)Bose, Dujmović, Hurtado, Langerman, Morin, and
  Wood}]{bose.dujmovic.ea:untangling}
Prosenjit Bose, Vida Dujmović, Ferran Hurtado, Stefan Langerman, Pat Morin,
  and David~R. Wood.
\newblock A polynomial bound for untangling geometric planar graphs.
\newblock \emph{Discret. Comput. Geom.}, 42(4):570--585, 2009.
\newblock \doi{10.1007/S00454-008-9125-3}.

\bibitem[{Brandenburg et~al.(2003)Brandenburg, Eppstein, Goodrich, Kobourov,
  Liotta, and Mutzel}]{DBLP:conf/gd/BrandenburgEGKLM03}
Franz{-}Josef Brandenburg, David Eppstein, Michael~T. Goodrich, Stephen~G.
  Kobourov, Giuseppe Liotta, and Petra Mutzel.
\newblock Selected open problems in graph drawing.
\newblock In Giuseppe Liotta, editor, \emph{Graph Drawing, 11th International
  Symposium, {GD} 2003, Perugia, Italy, September 21-24, 2003, Revised Papers},
  volume 2912 of \emph{Lecture Notes in Computer Science}, pages 515--539.
  Springer, 2003.
\newblock \doi{10.1007/978-3-540-24595-7\_55}.

\bibitem[{Bra{\ss} et~al.(2007)Bra{\ss}, Cenek, Duncan, Efrat, Erten,
  Ismailescu, Kobourov, Lubiw, and
  Mitchell}]{DBLP:journals/comgeo/BrassCDEEIKLM07}
Peter Bra{\ss}, Eowyn Cenek, Christian~A. Duncan, Alon Efrat, Cesim Erten, Dan
  Ismailescu, Stephen~G. Kobourov, Anna Lubiw, and Joseph S.~B. Mitchell.
\newblock On simultaneous planar graph embeddings.
\newblock \emph{Comput. Geom.}, 36(2):117--130, 2007.
\newblock \doi{10.1016/J.COMGEO.2006.05.006}.
\newblock Also in, 8th Int. Workshop on Algorithms and Data Structures
  {(WADS)}, pages 243--255. (2003).

\bibitem[{Cano et~al.(2014)Cano, T{\'{o}}th, and Urrutia}]{cano.toth.ea:upper}
Javier Cano, Csaba~D. T{\'{o}}th, and Jorge Urrutia.
\newblock Upper bound constructions for untangling planar geometric graphs.
\newblock \emph{{SIAM} J. Discret. Math.}, 28(4):1935--1943, 2014.
\newblock \doi{10.1137/130924172}.

\bibitem[{Cardinal et~al.(2015)Cardinal, Hoffmann, and
  Kusters}]{DBLP:journals/jgaa/CardinalHK15}
Jean Cardinal, Michael Hoffmann, and Vincent Kusters.
\newblock On universal point sets for planar graphs.
\newblock \emph{J. Graph Algorithms Appl.}, 19(1):529--547, 2015.
\newblock \doi{10.7155/JGAA.00374}.

\bibitem[{Casta{\~{n}}eda and Urrutia(1996)}]{DBLP:conf/cccg/CastanedaU96}
Netzahualcoyotl Casta{\~{n}}eda and Jorge Urrutia.
\newblock Straight line embeddings of planar graphs on point sets.
\newblock In \emph{Proceedings of the 8th Canadian Conference on Computational
  Geometry, {CCCG} 1996}, pages 312--318. 1996.

\bibitem[{Cibulka(2010)}]{cibulka:untangling}
Josef Cibulka.
\newblock Untangling polygons and graphs.
\newblock \emph{Discret. Comput. Geom.}, 43(2):402--411, 2010.
\newblock \doi{10.1007/S00454-009-9150-X}.

\bibitem[{{Da Lozzo} et~al.(2018){Da Lozzo}, Dujmović, Frati, Mchedlidze, and
  Roselli}]{dalozzo.dujmovic.ea:drawing}
Giordano {Da Lozzo}, Vida Dujmović, Fabrizio Frati, Tamara Mchedlidze, and
  Vincenzo Roselli.
\newblock Drawing planar graphs with many collinear vertices.
\newblock \emph{J. Comput. Geom.}, 9(1):94--130, 2018.
\newblock \doi{10.20382/JOCG.V9I1A4}.

\bibitem[{Diestel(2018)}]{Diestel5}
Reinhard Diestel.
\newblock \emph{Graph theory}, volume 173 of \emph{Graduate Texts in
  Mathematics}.
\newblock Springer, 5th edition, 2018.

\bibitem[{Dujmovi{\'c}(2017)}]{dujmovic:utility}
Vida Dujmovi{\'c}.
\newblock The utility of untangling.
\newblock \emph{J. Graph Algorithms Appl.}, 21(1):121--134, 2017.
\newblock \doi{10.7155/JGAA.00407}.
\newblock Also in, Graph Drawing and Network Visualization - 23rd International
  Symposium, {GD} 2015.

\bibitem[{Dujmovi{\'c} and Morin(2023)}]{dujmovic.morin:dual}
Vida Dujmovi{\'c} and Pat Morin.
\newblock Dual circumference and collinear sets.
\newblock \emph{Discret. Comput. Geom.}, 69(1):26--50, 2023.
\newblock \doi{10.1007/S00454-022-00418-4}.

\bibitem[{Dujmović et~al.(2021)Dujmović, Frati, Gon{\c{c}}alves, Morin, and
  Rote}]{dujmovic.frati.ea:every}
Vida Dujmović, Fabrizio Frati, Daniel Gon{\c{c}}alves, Pat Morin, and
  G{\"{u}}nter Rote.
\newblock Every collinear set in a planar graph is free.
\newblock \emph{Discret. Comput. Geom.}, 65(4):999--1027, 2021.
\newblock \doi{10.1007/S00454-019-00167-X}.

\bibitem[{Erdős(1946)}]{erdos:on}
Paul Erdős.
\newblock On sets of distances of n points.
\newblock \emph{The American Mathematical Monthly}, 53:248–250, 1946.

\bibitem[{Estrella{-}Balderrama et~al.(2009)Estrella{-}Balderrama, Fowler, and
  Kobourov}]{DBLP:journals/comgeo/Estrella-BalderramaFK09}
Alejandro Estrella{-}Balderrama, J.~Joseph Fowler, and Stephen~G. Kobourov.
\newblock Characterization of unlabeled level planar trees.
\newblock \emph{Comput. Geom.}, 42(6-7):704--721, 2009.
\newblock \doi{10.1016/J.COMGEO.2008.12.006}.

\bibitem[{Evans et~al.(2014)Evans, Kusters, Saumell, and
  Speckmann}]{DBLP:conf/gd/EvansKSS14}
William Evans, Vincent Kusters, Maria Saumell, and Bettina Speckmann.
\newblock Column planarity and partial simultaneous geometric embedding.
\newblock In \emph{Graph Drawing - 22nd International Symposium, {GD} 2014},
  volume 8871 of \emph{LNCS}, pages 259--271. Springer, 2014.
\newblock \doi{10.1007/978-3-662-45803-7\_22}.

\bibitem[{{Felsner} et~al.(2003){Felsner}, {Liotta}, and {Wismath}}]{JGAA-75}
{Stefan} {Felsner}, {Giuseppe} {Liotta}, and {Stephen} {Wismath}.
\newblock Straight-line drawings on restricted integer grids in two and three
  dimensions.
\newblock \emph{Journal of Graph Algorithms and Applications}, 7(4):363--398,
  2003.
\newblock \doi{10.7155/jgaa.00075}.

\bibitem[{de~Fraysseix et~al.(1988)de~Fraysseix, Pach, and
  Pollack}]{defraysseix.pach.ea:small}
Hubert de~Fraysseix, J{\'{a}}nos Pach, and Richard Pollack.
\newblock Small sets supporting {F}{\'{a}}ry embeddings of planar graphs.
\newblock In Janos Simon, editor, \emph{Proceedings of the 20th Annual {ACM}
  Symposium on Theory of Computing, May 2-4, 1988, Chicago, Illinois, {USA}},
  pages 426--433. {ACM}, 1988.
\newblock \doi{10.1145/62212.62254}.

\bibitem[{de~Fraysseix et~al.(1990)de~Fraysseix, Pach, and
  Pollack}]{defraysseix.pach.ea:how}
Hubert de~Fraysseix, J{\'{a}}nos Pach, and Richard Pollack.
\newblock How to draw a planar graph on a grid.
\newblock \emph{Comb.}, 10(1):41--51, 1990.
\newblock \doi{10.1007/BF02122694}.

\bibitem[{Fáry(1948)}]{fary:on}
István Fáry.
\newblock On straight-line representation of planar graphs.
\newblock \emph{Acta Sci. Math. (Szeged)}, 11:229--233, 1948.

\bibitem[{{D}i Giacomo et~al.(2023){D}i Giacomo, Didimo, Liotta, Meijer,
  Montecchiani, and Wismath}]{DBLP:journals/corr/abs-2311-14634}
Emilio {D}i Giacomo, Walter Didimo, Giuseppe Liotta, Henk Meijer, Fabrizio
  Montecchiani, and Stephen~K. Wismath.
\newblock New bounds on the local and global edge-length ratio of planar
  graphs.
\newblock \emph{CoRR}, abs/2311.14634, 2023.
\newblock \doi{10.48550/ARXIV.2311.14634}.
\newblock Also in, Proceedings of the 40th European Workshop on Computational
  Geometry, EuroCG 2024.

\bibitem[{{D}i Giacomo et~al.(2012){D}i Giacomo, Liotta, and
  Mchedlidze}]{DBLP:journals/corr/abs-1212-0804}
Emilio {D}i Giacomo, Giuseppe Liotta, and Tamara Mchedlidze.
\newblock How many vertex locations can be arbitrarily chosen when drawing
  planar graphs?
\newblock \emph{CoRR}, abs/1212.0804, 2012.
\newblock \eprint{1212.0804}.

\bibitem[{Giacomo et~al.(2009)Giacomo, Didimo, van Kreveld, Liotta, and
  Speckmann}]{DBLP:journals/jgaa/GiacomoDKLS09}
Emilio~Di Giacomo, Walter Didimo, Marc~J. van Kreveld, Giuseppe Liotta, and
  Bettina Speckmann.
\newblock Matched drawings of planar graphs.
\newblock \emph{J. Graph Algorithms Appl.}, 13(3):423--445, 2009.
\newblock \doi{10.7155/JGAA.00193}.

\bibitem[{Goaoc et~al.(2009)Goaoc, Kratochv{\'{\i}}l, Okamoto, Shin, Spillner,
  and Wolff}]{goaoc.kratochvil.ea:untangling}
Xavier Goaoc, Jan Kratochv{\'{\i}}l, Yoshio Okamoto, Chan{-}Su Shin, Andreas
  Spillner, and Alexander Wolff.
\newblock Untangling a planar graph.
\newblock \emph{Discret. Comput. Geom.}, 42(4):542--569, 2009.
\newblock \doi{10.1007/S00454-008-9130-6}.

\bibitem[{Goldner and Harary(1975)}]{goldner1975note}
Anita~M. Goldner and Frank Harary.
\newblock A note on the smallest nonhamiltonian maximal planar graph.
\newblock \emph{Bull. Malaysian Math. Soc.}, 6(1):41--42, 1975.

\bibitem[{Gritzmann et~al.(1991)Gritzmann, Mohar, P{\'a}ch, and Pollack}]{GMPP}
Peter Gritzmann, Bojan Mohar, Janos P{\'a}ch, and Richard Pollack.
\newblock Embedding a planar triangulation with vertices at specified points
  (solution to problem e3341).
\newblock \emph{Amer. Math. Monthly}, 98:165--166, 1991.

\bibitem[{Grünbaum and Walther(1973)}]{GRUNBAUM1973364}
Branko Grünbaum and Hansjoachim Walther.
\newblock Shortness exponents of families of graphs.
\newblock \emph{Journal of Combinatorial Theory, Series A}, 14(3):364--385,
  1973.
\newblock \doi{https://doi.org/10.1016/0097-3165(73)90012-5}.

\bibitem[{Jackson(1986)}]{jackson:longest}
Bill Jackson.
\newblock Longest cycles in 3-connected cubic graphs.
\newblock \emph{J. Comb. Theory {B}}, 41(1):17--26, 1986.
\newblock \doi{10.1016/0095-8956(86)90024-9}.

\bibitem[{Kang et~al.(2011)Kang, Pikhurko, Ravsky, Schacht, and
  Verbitsky}]{kang.pikhurko.ea:untangling}
Mihyun Kang, Oleg Pikhurko, Alexander Ravsky, Mathias Schacht, and Oleg
  Verbitsky.
\newblock Untangling planar graphs from a specified vertex position - hard
  cases.
\newblock \emph{Discret. Appl. Math.}, 159(8):789--799, 2011.
\newblock \doi{10.1016/J.DAM.2011.01.011}.

\bibitem[{Kant and Bodlaender(1997)}]{kant.bodlaender:triangulating}
Goos Kant and Hans~L. Bodlaender.
\newblock Triangulating planar graphs while minimizing the maximum degree.
\newblock \emph{Inf. Comput.}, 135(1):1--14, 1997.
\newblock \doi{10.1006/inco.1997.2635}.

\bibitem[{Liu et~al.(2018)Liu, Yu, and Zhang}]{liu.yu.zhang:circumference}
Qinghai Liu, Xingxing Yu, and Zhao Zhang.
\newblock Circumference of 3-connected cubic graphs.
\newblock \emph{J. Comb. Theory, Ser. {B}}, 128:134--159, 2018.
\newblock \doi{10.1016/j.jctb.2017.08.008}.

\bibitem[{Owens(1984)}]{owens:regular}
Peter~J. Owens.
\newblock Regular planar graphs with faces of only two types and shortness
  parameters.
\newblock \emph{J. Graph Theory}, 8(2):253--275, 1984.
\newblock \doi{10.1002/JGT.3190080207}.

\bibitem[{Pach and Tardos(2002)}]{pach.tardos:untangling}
J{\'{a}}nos Pach and G{\'{a}}bor Tardos.
\newblock Untangling a polygon.
\newblock \emph{Discret. Comput. Geom.}, 28(4):585--592, 2002.
\newblock \doi{10.1007/S00454-002-2889-Y}.
\newblock Also in, Graph Drawing, 9th International Symposium, {GD} 2001.

\bibitem[{Pach and T{\'{o}}th(2004)}]{pach.toth:monotone}
J{\'{a}}nos Pach and G{\'{e}}za T{\'{o}}th.
\newblock Monotone drawings of planar graphs.
\newblock \emph{J. Graph Theory}, 46(1):39--47, 2004.
\newblock \doi{10.1002/JGT.10168}.

\bibitem[{Ravsky and Verbitsky(2011)}]{DBLP:conf/wg/RavskyV11}
Alexander Ravsky and Oleg Verbitsky.
\newblock On collinear sets in straight-line drawings.
\newblock In Petr Kolman and Jan Kratochv{\'{\i}}l, editors,
  \emph{Graph-Theoretic Concepts in Computer Science - 37th International
  Workshop, {WG} 2011, Tepl{\'{a}} Monastery, Czech Republic, June 21-24, 2011.
  Revised Papers}, volume 6986 of \emph{Lecture Notes in Computer Science},
  pages 295--306. Springer, 2011.
\newblock \doi{10.1007/978-3-642-25870-1\_27}.

\bibitem[{Robertson et~al.(1997)Robertson, Sanders, Seymour, and
  Thomas}]{robertson.seymour.ea:four-colour}
Neil Robertson, Daniel~P. Sanders, Paul~D. Seymour, and Robin Thomas.
\newblock The four-colour theorem.
\newblock \emph{J. Comb. Theory, Ser. {B}}, 70(1):2--44, 1997.
\newblock \doi{10.1006/jctb.1997.1750}.

\bibitem[{Scheucher et~al.(2020)Scheucher, Schrezenmaier, and
  Steiner}]{scheucher.schrezenmaier.ea:note}
Manfred Scheucher, Hendrik Schrezenmaier, and Raphael Steiner.
\newblock A note on universal point sets for planar graphs.
\newblock \emph{J. Graph Algorithms Appl.}, 24(3):247--267, 2020.
\newblock \doi{10.7155/JGAA.00529}.

\bibitem[{Schnyder(1990)}]{schnyder:embedding}
Walter Schnyder.
\newblock Embedding planar graphs on the grid.
\newblock In David~S. Johnson, editor, \emph{Proceedings of the First Annual
  {ACM-SIAM} Symposium on Discrete Algorithms, 22-24 January 1990, San
  Francisco, California, {USA}}, pages 138--148. {SIAM}, 1990.

\bibitem[{Stein(1951)}]{stein:convex}
Sherman~K. Stein.
\newblock Convex maps.
\newblock \emph{Proceedings of the American Mathematical Society},
  2(3):464--466, 1951.
\newblock \doi{10.2307/2031777}.

\bibitem[{Steiner(2023)}]{DBLP:conf/gd/Steiner23}
Raphael Steiner.
\newblock A logarithmic bound for simultaneous embeddings of planar graphs.
\newblock In \emph{Graph Drawing and Network Visualization - 31st International
  Symposium, {GD} 2023}, volume 14466 of \emph{LNCS}, pages 133--140. Springer,
  2023.
\newblock \doi{10.1007/978-3-031-49275-4\_9}.

\bibitem[{Tait(1880)}]{tait:remarks}
Peter~Guthrie Tait.
\newblock Remarks on the colouring of maps.
\newblock \emph{Proc. Roy. Soc. Edinburgh Sect. A}, 10:729, 1880.

\bibitem[{Tutte(1946)}]{tutte:on}
William~T. Tutte.
\newblock On {H}amilton circuits.
\newblock \emph{J. Lond. Math. Soc.}, 21:98--101, 1946.

\bibitem[{Tutte(1963)}]{tutte:how}
William~T. Tutte.
\newblock How to draw a graph.
\newblock \emph{Proceedings of The London Mathematical Society}, 13:743--767,
  1963.

\bibitem[{Wagner(1936)}]{wagner:bemerkungen}
Klaus Wagner.
\newblock Bemerkungen zum {V}ierfarbenproblem.
\newblock \emph{Jahresbericht der Deutschen Mathematiker-Vereinigung},
  46:26--32, 1936.

\bibitem[{Woodall(1971)}]{woodall:thrackles}
Douglas~R Woodall.
\newblock Thrackles and deadlock.
\newblock \emph{Combinatorial Mathematics and Its Applications}, 348:335–348,
  1971.

\end{thebibliography}

\end{document}